\documentclass[10pt]{llncs}
\usepackage{amssymb,amsmath}
\usepackage{enumitem}
\usepackage{cite}
\usepackage{xcolor,ntheorem,complexity}
\usepackage{graphicx}
\usepackage[breaklinks=true]{hyperref}
\colorlet{green}{green!50!black} 
\hypersetup{
	colorlinks=true,
	citecolor=green,
	linkcolor=blue,  
	urlcolor=black,
}
\usepackage{algorithm}
\usepackage{algorithmicx}
\usepackage[noend]{algpseudocode}
\usepackage[T1]{fontenc}

\newcommand{\myproblem}[1]{{\textsc{#1}}} 
\def\OPT{\textup{OPT}}
\def\bp{\textup{bp}}
\def\S{\mathcal{S}}
\def\FamML{\mathcal{F}_{\textup{ML}}}
\def\NP{\mathsf{NP}}
\def\I{\mathcal{I}}
\def\J{\mathcal{J}}
\def\ml{\textup{ml}}
\def\top{\textup{top}}
\def\obj{\textup{obj}}
\def\OPT{\textup{OPT}}

\def\MUPMU{\myproblem{Max-Utility Popular Matching with Instability Costs}}
\def\FAS{\myproblem{Feedback Arc Set}}
\def\HS{\myproblem{Hitting Set}}
\def\distswap{\Delta^\textup{swap}}
\def\distedge{\Delta^\textup{edge}}
\def\distvert{\Delta^\textup{vert}}

\newtheorem{observation}[theorem]{Observation}
\newtheorem{mycorollary}[theorem]{Corollary}
\newtheorem{mylemma}[theorem]{Lemma}

\title{Recognizing when a preference system is close to admitting a master list}

\author{Ildik\'{o} Schlotter\inst{1,2}}
\institute{
 Centre for Economic and Regional Studies, Budapest, Hungary; \email{schlotter.ildiko@krtk.hu} 
\and
Budapest University of Technology and Economics, Budapest, Hungary
}

\begin{document}

\maketitle

\begin{abstract}
A preference system~$\I$ is an undirected graph where vertices have preferences over their neighbors, 
and $\I$ admits a master list if all preferences can be derived from a single ordering over all vertices.
We study the problem of deciding whether a given preference system~$\I$ is \emph{close to} admitting a master list
based on three different distance measures. We determine the computational complexity of the following questions: 
can $\I$ be modified by (i) $k$ swaps in the preferences, (ii) $k$ edge deletions, 
or (iii)  $k$ vertex deletions 
so that the resulting instance admits a master list? 
We investigate these problems in detail from the viewpoint of parameterized complexity and of approximation.
We also present two applications related to stable and popular matchings. 
\end{abstract}

\section{Introduction}
\label{sec:intro}

A preference system models a set of agents as an undirected graph where agents are vertices, 
and each agent has preferences over its neighbors. 
Preference systems are a fundamental concept in the area of matching under preferences which,
originating in the seminal work of Gale and Shapley~\cite{gale-shapley} on stable matchings, 
is a prominent research field in the intersection of algorithm design and computational social choice
that has steadily gained attention over the last two decades.

Preference systems may admit a master list, that is, a global ranking over all agents from which agents derive their preferences.
Master lists arise naturally in many practical scenarios such as P2P networks~\cite{GLMMRV07,LMVGRM07}, 
job markets~\cite{IMS08}, and student housing assignments~\cite{PPR08}.
Consequently, master lists and its generalizations have been the focus of research in 
several papers~\cite{IMS08,KW16,Kam19,BHK+20,MR20,CR21,Chowdhury22}.

In this work we aim to investigate the computational complexity of recognizing preference systems that are \emph{close to} admitting a master list. 
Such instances may arise as a result of noise in the data set, 
or in scenarios where a global ranking of agents is used in general, with the exception of a few anomalies.

\subsection{Our contribution}
We introduce three measures to describe the distance of a given preference system~$\I$ from the class of preference systems
admitting a master list. The first measure, $\distswap(\I)$ is based on the swap distance between agents' preferences, 
while the measures~$\distedge(\I)$ and $\distvert(\I)$ are based on classic graph operations,  
the deletion of edges or vertices; precise definitions follow in Section~\ref{sec:prelim}.
We study in detail the complexity of computing these values for a given preference system~$\I$.
After proving that computing any of these three measures is $\NP$-hard, 
we apply the framework of parameterized complexity and of approximation algorithms to gain a more fine-grained insight.

In addition to the problems of computing $\distswap(\I)$, $\distedge(\I)$, and $\distvert(\I)$, 
we briefly look at two applications.
First, we show that if a strict preference system~$\I$ is close to admitting a master list, 
then we can bound the number of stable matchings as a function of the given distance measure. 
This yields an efficient way to solve a wide range of stable matching problems in instances that are close to admitting a master list. 
Second, we consider an optimization problem over popular matchings where the task is to find a maximum-utility popular matching
while keeping the number (or cost) of blocking edges low.
We prove that this notoriously hard problem can be efficiently solved 
if preferences are close to admitting a master list. 
In both of these applications, the running time of the obtained algorithms heavily depends on the distance measure used.

\subsection{Related work}
Master lists have been extensively studied in the context of stable matchings~\cite{IMS08,Kam19,BHK+20,Chowdhury22}. 
Various models have been introduced in the literature to generalize master lists, and 
capture preferences that are similar to each other in some sense.
Closest to our work might be the paper by 
Bredereck et al.~\cite{BHK+20}
who examine the complexity of multidimensional stable matching problems on instances that are, in some sense, close to admitting a master list. 
Abraham et al. investigated a setting where agent pairs are ranked globally~\cite{ALMO08}. 
Bhatnagar et al.~\cite{BGR08} examined three restrictions on preference systems: 
the $k$-attribute model where agents are evaluated through a linear function of $k$ attributes, 
the $k$-range model where for each agent~$p$ on one side of a bipartite preference system the rankings of~$p$ 
by the agents on the other side falls into a range of size~$k$, and 
the $k$-list model where agents can be partitioned into $k$ types, 
and agents of the same type have identical preferences.
Meeks and Rastegari studied hard problems under a variant of the $k$-list model~\cite{MR20}. 
Khanchandani and Wattenhofer considered a distributed stable matching problem in a bipartite $k$-range model~\cite{KW16}.
Cheng and Rosenbaum~\cite{CR21} examined the expressive powers of the $k$-range, $k$-attribute, and a bipartite version of the $k$-list model
in terms of the rotation poset realized by such preference systems. 

Restricted preference profiles have been also extensively studied in the broader context of computational social choice;
see the survey by Elkind et al.~\cite{ELP17}. 
In election systems, computing the Kemeny score~\cite{Kemeny59} for a multiset of votes 
(where each vote is a total linear order over a set of candidates)
is analogous to computing the value~$\distswap(\I)$ for a preference system~$\I$, 
although there are some differences between these two problems. 
Apart from the literature on the computational complexity of computing the Kemeny score (see e.g.~\cite{BTT89,HSV05,BFGNR09,FH21}),
our work also relates to the problem of computing certain distance measures between elections
which has been investigated by Boehmer et al.~\cite{BFNSW22}.
Distance from a strongly unanimous profile (where all votes coincide) has been considered as a 
parameter in the study by Gupta et al. on committee selection~\cite{GJS20}.

\subsubsection{Organization.}
We start with the preliminaries in Section~\ref{sec:prelim},
and define our three distance measures and the corresponding problems in Section~\ref{sec:initial_results}.
Section~\ref{sec:main-res} contains our main results, first for strict and then for weak preferences, 
and we present our two applications in Section~\ref{sec:applic}. We conclude with a short summary 
and some suggestions for future research in Section~\ref{sec:conclusion}.

\section{Preliminaries}
\label{sec:prelim}

Let us now introduce our notation, giving all necessary definition about directed
and undirected graphs, preference systems, and certain concepts of computational complexity theory.

\subsection{Graph-theoretic notions}

We define most of the basic graph-related concepts we need for directed graphs,
as the corresponding concepts for undirected graphs are usually defined analogously.
We use the notation $[n]=\{1,\dots,n\}$ for any $n \in \mathbb{Z}^+$.

\medskip 
\noindent
{\bf Directed graphs.}
For a directed graph (or \emph{digraph})~$D$, we denote by $V(D)$ and $A(D)$ its vertex and arc set, respectively; 
each arc is an ordered pair of vertices.
For some vertex~$v \in V(G)$ we let $N^-_D(v)$ and $N^+_D(v)$ denote 
the set containing $v$'s \emph{in-} and \emph{out-neighbors} in~$D$, respectively. 
A \emph{loop} is an arc~$(v,v)$ for some~$v \in V(D)$, 
and two arcs are \emph{parallel} if they both point from~$a$ to~$b$ for some $a,b \in V(D)$.
In this paper, 
parallel arcs will be labelled, and we will be able to identify them via their labels.
A \emph{source} in~$D$ is a vertex with in-degree~0, and a \emph{sink} is one with out-degree~0.
A \emph{path} in~$D$ is a series $v_1, \dots, v_p$ of distinct vertices such that $(v_i,v_{i+1}) \in A(D)$ for each $i \in [p-1]$;
if additionally, $(v_p,v_1)$ is also an arc in~$D$, then these vertices form a \emph{cycle} in~$D$.
We say that~$D$ is acyclic, if it contains no cycles. The vertices of an acyclic digraph can be listed in a \emph{topological order}
which is a strict order $v_1, \dots, v_n$ of the vertices of~$D$ such that any arc~$(v_i,v_j) \in A(D)$ satisfies $i <j$. 

\medskip
\noindent
{\bf Graph operations.}
A \emph{subgraph} of $D$ is a graph obtained by deleting some edges and vertices from~$D$. 
For a set~$X$ of edges or vertices in~$D$, we let $D-X$ denote the subgraph of~$D$ 
obtained by deleting~$X$ from~$D$. For a set~$W \subseteq V(D)$ of vertices, we let $D[W]=D-(V(D)\setminus W)$ denote the subgraph \emph{induced} by~$W$. 
We may interpret paths or cycles in~$G$ also as subgraphs of~$G$. 
By \emph{subdividing} an arc~$(u,v)$ we mean replacing~$(u,v)$ with two new arcs~$(u,x)$ and $(x,v)$ 
where $x$ is a newly introduced vertex; the inverse operation is \emph{lifting} the vertex $x$ from the path~$(u,x,v)$.
By \emph{contracting} a subgraph $D'$ of~$D$ we mean the following: we replace $V(D')$ with a new vertex $x$, 
and for each arc~$e \in A(D)$ we replace any endpoint of~$e$ contained in~$D'$ by~$x$; if any loops are created in the process, we delete them.

\medskip 
\noindent
{\bf Undirected graphs.}
Given an undirected graph~$G$, we let $V(G)$ and $E(G)$ denote its vertex and edge set, respectively;
each edge is a set of two of vertices.
For some vertex~$v \in V(G)$ we let $N_G(v)$ denote the set containing $v$'s \emph{neighbors} in~$G$. 
All other concepts defined in the previous two paragraphs are defined analogously for undirected graphs.
A \emph{clique} is a set $K \subseteq V(G)$ of vertices such that there is an edge between each two vertices of~$K$. 
The \emph{bidirected} version of~$G$ is obtained by replacing each edge~$\{u,v\} \in E(G)$ by two arcs~$(u,v)$ and $(v,u)$. 

\subsection{Preference systems}

A \emph{preference system} is defined as a pair $\I=(G,\preceq)$ where $G$ is an undirected graph and $\preceq=\{\preceq_v:v \in V(G)\}$ 
where $\preceq_v$ is a weak or a strict order over~$N_G(v)$ for each vertex~$v \in V(G)$, indicating the \emph{preferences} of~$v$.
For some $v \in V(G)$ and $a,b \in N_G(v)$, we say that $v$ \emph{prefers}~$b$ to~$a$, 
denoted by $a \prec_v b$, if $a \not \preceq_v b$.
We write $a \sim_v b$, if $a \preceq_v b$ and $b \preceq_v a$.
A \emph{tie} in~$v$'s preferences is a maximal set~$T \subseteq N_G(v)$ such that $t \sim_v t'$ for each~$t$ and~$t'$ in~$T$.
If each tie has size~1, then $\I$ is a \emph{strict preference system}, 
and we may denote it by $(G,\prec)$. 

\medskip
\noindent
{\bf Deletions and Swaps.}
For a set~$X$ of edges or vertices in~$G$, let 
$\I-X$ denote the preference system whose underlying graph is~$G-X$ and 
where the preferences of each vertex $v \in V(G-X)$ is the restriction of~$\preceq_v$ to $N_{G-X}(v)$.
We may refer to~$\I-X$ as a sub-instance of~$\I$.

If vertex $v$ has strict preferences $\prec_v$ in~$\I$, then a \emph{swap} is a triple $(a,b;v)$ with~$a,b \in N_G(v)$, 
and it is \emph{admissible} if $a$ and $b$ are consecutive\footnote{Vertices $a$ and $b$ are consecutive in $v$'s preferences, if either $a \prec_v b$ but there is no vertex~$c$ with $a \prec_v c \prec_v b$, or $b \prec_v a$ but there is no vertex~$c$ with $b \prec_v c \prec_v a$.} in~$v$'s preferences.
\emph{Performing} an admissible swap~$(a,b;v)$ in~$\I$ means switching $a$ and $b$ in~$v$'s preferences; 
the resulting preference system is denoted by~$\I \lhd (a,b;v)$. 
For a set~$S$ of swaps, $\I \lhd S$ denotes the preference system obtained by performing the swaps in~$S$ in~$\I$ 
in an arbitrary order as long as each swap is admissible (if this is not possible, $\I \lhd S$ is undefined).
For non-strict preferences, similar notions will be discussed in Section~\ref{sec:initial_results}.

\medskip
\noindent
{\bf Master Lists.}
A weak or strict order $\preceq^\ml$ over $V(G)$ is a \emph{master list} for $(G,\preceq)$, if for each $v \in V(G)$, 
the preferences of~$v$ are \emph{consistent} with~$\preceq^\ml$, that is, $\preceq_v$ is the restriction of~$\preceq^\ml$ to~$N_G(v)$. 
We will denote by $\FamML$ the family of those preference systems that admit a master list.
Notice that $\FamML$ is closed under taking subgraphs: if we
delete a vertex or an edge from a preference system in~$\FamML$, the remainder still admits a master list.

\medskip
\noindent
{\bf Stable and popular matchings.}
A \emph{matching} in a preference system~$\I$ is a set of edges such that no two of them share an endpoint.
For an edge $\{a,b\}$ in~$M$, we say that $M$ \emph{matches} $a$ to~$b$ and vice versa, and we denote this by setting
$M(a)=b$ and $M(b)=a$. A vertex~$v$ in~$G$ is \emph{unmatched} in~$M$ if no edge of~$M$ has~$v$ as its endpoint.
A blocking edge for~$M$ is an edge $\{a,b\}$ of~$G$ such that (i)~either $a$ is unmatched in~$M$ or $M(a) \prec_a b$, 
and (ii)  either $b$ is unmatched in~$M$ or $M(b) \prec_b a$.
The set of blocking edges for~$M$ is denoted by~$\bp(M)$, and $M$ is \emph{stable} if $\bp(M)=\emptyset$.
For a preference system~$\I$, we will denote by~$\mathcal{S}(\I)$ the set of stable matchings in~$\I$.

We say that a vertex~$v$ prefers some matching~$M$ over another matching~$M'$, if either $v$ is unmatched in~$M'$ but matched in~$M$, or
it is matched in both but $M'(v) \prec_v M(v)$.
For two matchings $M$ and $M'$ in~$\I$, we say that $M$ is \emph{more popular than} $M'$, 
if the number of vertices that prefer $M$ to~$M'$ is more than the number of vertices preferring~$M'$ to~$M$. 
A matching~$M$ in \emph{popular}  in~$\I$, if there is no matching more popular than~$M$.

\subsection{Computational complexity}

We assume that the reader is familiar with the basic concepts and tools of
classic computational complexity theory, such as $\NP$-hardness and polynomial-time reductions.

\medskip
\noindent
{\bf Parameterized complexity.}
In the framework of parameterized complexity, each problem instance~$\I$ is associated with an integer~$k$,
and the aim is to find algorithms whose running time is $f(k) |\I|^{O(1)}$ for some computable function~$f$;
such an algorithm is called \emph{fixed-parameter tractable} (or FPT) with parameter~$k$. 
The basic complexity class in the parameterized framework is $\mathsf{W}[1]$, and if a parameterized problem  is 
proved to be $\mathsf{W}[1]$-hard, then this is considered a strong evidence that it is not FPT with the given parameterization. 
To show that a given parameterized problem $\mathcal{P}$ is $\mathsf{W}[1]$-hard, it suffices to give a \emph{parameterized} or \emph{FPT reduction} 
from another $\mathsf{W}[1]$-hard problem~$\mathcal{Q}$, 
which is a function $f$ that for each instance $(\I,k)$ of the problem~$\mathcal{Q}$ 
computes in FPT time an equivalent instance $(\I',k')$ of~$\mathcal{P}$ such that $k' \leq g(k)$ for some function~$g$.

\medskip
\noindent
{\bf Approximation.}
For a minimization problem $\mathcal{P}$, let $\obj_{\mathcal{P}}$ denote its objective function,
and let $\textup{OPT}_{\mathcal{P}}(\I)=\obj_\I(S)$ denote the value of an optimal solution~$S$ 
for an instance~$\I$ of~$\mathcal{P}$.
For some $c \geq 1$, an algorithm is a \emph{$c$-approximation} or an approximation with \emph{factor}~$c$ for $\mathcal{P}$,
if for any instance~$\I$ of~$\mathcal{P}$ it returns a solution~$S$ for~$\I$ with $\obj_\mathcal{P}(S) \leq c \cdot \OPT_{\mathcal{P}}(I)$.
For two minimization problems~$\mathcal{P}_1$ and~$\mathcal{P}_2$, 
an \emph{approximation preserving reduction}~\cite[Section~A.3.1]{Vazirani-book}
from~$\mathcal{P}_1$ to~$\mathcal{P}_2$
consists of two functions~$f$ and~$g$ computable in polynomial time such that
\begin{itemize}
\item[(i)] for any instance $\I_1$ of $\mathcal{P}_1$, function~$f$ yields an instance $\I_2=f(\I_1)$ of $\mathcal{P}_2$ 
with $\OPT_{\mathcal{P}_2}(\I_2) \leq \OPT_{\mathcal{P}_1}(\I_1)$, and 
\item[(ii)] for any solution $S_2$ of~$\I_2$, function~$g$ yields a solution~$S_1=g(S_2)$ for~$\I_1$ 
with~$\obj_{\I_1}(S_1) \leq \obj_{\I_2}(S_2)$.
\end{itemize}
Given such a reduction, an approximation algorithm~$\mathcal{A}$ for~$\mathcal{P}_2$ with factor~$c$ 
yields an approximation algorithm for~$\mathcal{P}_1$ with the same factor.

\smallskip 
Some of our intractability results hold under certain standard complexity-theoretic assumptions such as~$\mathsf{W}[1]\neq \mathsf{FPT}$
or the so-called \emph{Unique Games Conjecture}~\cite{Khot-UCG-02} which we will use in our paper without giving its precise definition.

For more on these topics, we refer the reader to the books~\cite{garey-johnson-book,downey-fellows-FPC-book,CyganEtAl2015,Vazirani-book}.

\section{Problem definition and initial results}
\label{sec:initial_results}

In Section~\ref{sec:characterization} we introduce the notion of a \emph{preference digraph}, 
a directed graph associated with a given preference system, which can 
be exploited to obtain a useful characterization of preference systems that admit a master list. 
In Section~\ref{sec:dist-def} we define our three measures for describing the distance 
of a preference system from~$\FamML$, and observe a simple fact regarding the relationship between these distances.


\subsection{Characterization of $\FamML$ through the preference diraph}
\label{sec:characterization}
%
With a strict preference system $\I=(G,\prec)$ where $G=(V,E)$, we associate an arc-labelled directed graph~$D_\I$ 
that we call the \emph{preference digraph} of $\I$. We let~$D_\I$ have the same set of vertices as~$G$, 
and we define the arcs in~$D_\I$ by adding an arc $(a,b)$ labelled with~$v$ whenever $a \prec_v b$ holds for some vertices~$a,b$ and $v$ in~$V$.
Note that several parallel arcs may point from~$a$ to~$b$ in~$D_\I$, each having a different label, 
so we have $|V(D_\I)|=|V|$ but $|A(D_\I)|=O(|V|\!\cdot\!|E|)$.
Observation~\ref{obs:pref-digraph} immediately follows from the fact that acyclic digraphs admit a topological order.

\begin{observation}
\label{obs:pref-digraph}
A strict preference system $(G,\prec)$ admits a master list if and only if the preference digraph of~$G$ is acyclic.
\end{observation}

For a preference system $\I=(G,\preceq)$ with $G=(V,E)$ that is not necessarily strict
we extend the concept of the preference digraph of~$\I$ as follows.
Again, we let $D_\I$ have $V$ as its vertex set, but now we add two types of arcs to~$D_\I$: 
for any $v$ in~$V$ and $a,b \in N_G(V)$ with $a \neq b$ 
we add a \emph{strict arc} $(a,b)$ with label~$v$ whenever $a \prec_v b$, 
and we add a pair of \emph{tied arcs}~$(a,b)$ and $(b,a)$, both with label~$v$, whenever $a \sim_v b$.
Note that this way we indeed generalize our definition above for the preference digraph of strict preference systems.
We will call a cycle of~$D_\I$ that contains a strict arc a \emph{strict cycle}.
The following lemma is a straightforward generalization of Observation~\ref{obs:pref-digraph}.

\begin{mylemma}
\label{lem:pref-digraph-ties}
A preference system $(G,\preceq)$ admits a master list if and only if no cycle of the preference digraph of~$G$ is strict.
\end{mylemma}
\begin{proof}
By definition, $(G,\preceq) \in \FamML$ if and only if there exists a weak order~$\preceq^{\ml}$ on~$V$
such that $\preceq_v$ is the restriction of~$\preceq^\ml$ to $N_G(v)$ for any~$v \in V$.

Let $C$ be a cycle in $D_\I$, and suppose that a master list $\preceq^\ml$ exists for~$\I$.
Note that for any arc~$(a,b)$ in~$C$ we know that $a \preceq_v b$ for some $v \in V$,
implying $a \preceq^\ml b$. As this holds for any arc on the cycle, we get that only $a \sim^{\ml} b$ is possible, 
and therefore no arc on~$C$ can be strict. 

For the other direction, suppose that no cycle in~$D_\I$ contains a strict arc. 
Let us define a master list~$\preceq^\ml$ on~$V$ as follows. First, we let $a \sim^\ml b$ whenever $a$ is reachable from $b$ via tied arcs. 
Clearly, $\sim^\ml$ is symmetric, and moreover, if $a \sim^\ml b$, then $a \prec_v b$ is not possible for any $v \in V$, 
as otherwise a path from~$b$ to~$a$ via tied arcs would form a cycle with the strict arc $(a,b)$ with label~$v$.
Furthermore, contracting all tied arcs creates an acyclic digraph~$D'$.
Let $\prec^\top$ be a topological ordering on the vertices of~$D'$.
For each $v \in V$, let $\varphi(v)$ denote the vertex of $D'$ to which $v$ has been contracted.
We set $a \prec^\ml b$ whenever $\varphi(a) \prec^\top \varphi(b)$ holds. 
It is easy to see that $\prec^\ml$ is a weak ordering on~$V$.

To see that $\prec^\ml$ is a master list for~$\I$, let $a$ and $b$ be two distinct vertices in~$V$.
First note that $a \sim_v b$ for some $v \in V$ immediately implies $a \sim^\ml b$. 
Second, if $a \prec_v b$ for some $v \in V$, then we already know $a \not\sim^\ml b$, yielding $\varphi(a) \neq \varphi(b)$. 
Therefore, $D'$ contains an arc from $\varphi(a)$ to $\varphi(b)$, yielding $\varphi(a) \prec^\top \varphi(b)$ and hence $a \prec^\ml b$. 
\qed
\end{proof}

\subsection{Measuring the distance from $\FamML$}
\label{sec:dist-def}

Let us now define our three measures to describe the distance of a given strict preference system~$\I=(G,\prec)$ from the family $\FamML$ of preference systems that admit a master list:
\begin{itemize}
\item $\distswap(\I)=\min\{|S|:S$ is a set of swaps in~$\I$ such that $\I \lhd S \in \FamML\}$;
\item $\distedge(\I)=\min\{|S|:S \subseteq E(G), \I-S \in \FamML\}$;
\item $\distvert(\I)=\min\{|S|:S \subseteq V(G), \I-S \in \FamML\}$.
\end{itemize}
The measures $\distedge(\I)$ and $\distvert(\I)$ can be easily extended for preference systems that are not necessarily strict, 
since the above definitions are well-defined for any preference system~$(G,\preceq)$. 

Extending the measure $\distswap(\I)$ for non-strict preference systems is, however, not entirely straightforward.
If there are ties in the preferences of some vertex~$v$, how can we define an admissible swap? 
In this paper we use the following definition for swap distance, which seems to be standard in the literature~\cite{BCKLN20,CSS21}.
Let $\preceq_u$ and $\preceq_v$ be weak orders. If they are not defined over the same sets, then the
\emph{swap distance} of $\preceq_u$ and $\preceq_v$, denoted by $\Delta(\preceq_u,\preceq_v)$ is $\infty$, 
otherwise 
$$
\Delta(\preceq_u,\preceq_v) = |\{\{a,b\}: a \prec_u b \textrm{ but } b \preceq_v a \}| + |\{\{a,b\}: a \sim_u b \textrm{ but } a \not\sim_v b \}|.$$
For two preferences systems~$\I=(G,\preceq)$ and $\I'=(G',\preceq')$ with $G=(V,E)$ and $G'=(V',E')$, we 
let their swap distance, denoted by $\Delta(\I,\I')$, be $\infty$ if they are not defined over the same vertex set;
otherwise (that is, if $V=V'$) we let $\Delta(\I,\I')=\sum_{v \in V} \Delta(\preceq_v,\preceq'_v)$.
Using this, we can define
$$\distswap(\I)= \min \left\{ \Delta(\I,\I'): \I' \in \FamML \right\} .$$

The following lemma follows easily from the definitions.

\begin{mylemma}
\label{prop:comparing-modulators}
$\distswap(\I) \geq \distedge(\I) \geq \distvert(\I)$ for any preference system~$\I$.
\end{mylemma}
\begin{proof}
Suppose that $\distswap(\I) = k$, meaning that $\sum_{v \in V} \Delta(\preceq_v,\preceq'_v) \leq k$ 
for some instance $\I'=(G,\preceq') \in \FamML$.
Let $U$ be the family containing all pairs of the form $(\{a,b\},v)$ for which 
$a \prec_v b$  but $b \preceq'_v a$, or $a \sim_v b$ but $a \not\sim'_v b$. 
By definition, $|U|=\sum_{v \in V} \Delta(\preceq_v,\preceq'_v) \leq k$.
Note that for each $(\{a,b\},v) \in U$, deleting the edge~$e=\{a,v\}$ from~$G$ 
decreases the distance of $\I$ from~$\I'$, that is, $\Delta(\I-e,\I'-e) \leq k-1$
(clearly, the same holds for the edge $\{b,v\}$ as well). 
Repeating this for each pair in~$U$, we can delete a set $S$ of at most~$k$ edges from~$G$ 
so that $\Delta(\I-S,\I'-S)=0$. Since $\FamML$ is closed under edge deletions, we know $\I'-S \in \FamML$, 
implying $\distedge(\I) \leq |S| \leq k$.

Proving the remaining statement of the lemma is easier: 
instead of deleting $k=\distedge(\I)$ edges to obtain an instance~$\I' \in \FamML$, 
we can simply delete a set of at most~$k$ vertices covering these edges in~$G$ to get a sub-instance of~$\I'$.
This shows $\distvert(\I) \leq \distedge(\I)$.
\qed
\end{proof}

Let \myproblem{Master List by Swaps} (or MLS for short) be the problem whose input is a preference system~$\I$ and an integer~$k$, 
and the task is to decide whether $\distswap(\I) \leq k$. 
Assuming that $\I$ is a strict preference system, a set $S$ of at
most~$k$ swaps whose application in $\I$ yields an instance admitting a master list
is called a \emph{solution} for the instance $(\I,k)$ of MLS. 
We define the \myproblem{Master List by Edge Deletion} (or MLED) and the \myproblem{Master List by Vertex Deletion} (or MLVD) 
problems and their solution concepts 
analogously.

\subsubsection{MLS versus the \myproblem{Kemeny Score} problem.}
We note that the MLS problem for a strict preference system~$(G,\prec)$ on a graph $G=(V,E)$ can be reformulated 
as an instance of the \myproblem{Kemeny Score} problem where votes are allowed to be incomplete, 
by setting $V$ as the candidate set and interpreting 
each $\prec_v$, $v \in V$, as a vote containing only the candidates~$N_G(v)$;
we refer to Betzler et al.~\cite{BFGNR09} for a precise definition of this variant of \myproblem{Kemeny Score}.
Note that even if $G$ is a complete graph, $N_G(v)$ excludes~$v$, and 
thus MLS differs from the classic form of the \myproblem{Kemeny Score} problem where each vote is a ranking over the whole candidate set.
We further remark that, when considering weak orders, MLS and \myproblem{Kemeny Score} significantly differ in the way they treat ties~\cite{HSV05,BFGNR09}. 
We will comment in more detail on the connection between our findings for MLS and the known results for \myproblem{Kemeny Score} in Section~\ref{sec:main-res-strict}.

\section{Computing the distance from admitting a master list}
\label{sec:main-res}

Let us now present our main results on recognizing when a given preference list is close to admitting a master list.
We investigate the classical and parameterized complexity of the problem of determining 
each of the three distance measures defined in Section~\ref{sec:prelim} for a given preference system, 
that is, the problems MLS, MLED, and MLVD. 
In Section~\ref{sec:main-res-strict} we consider strict preference systems, 
and then extend our results for weakly ordered preferences in Section~\ref{sec:main-res-ties}.

\subsection{Strict preferences}
\label{sec:main-res-strict}


We are going to show that computing the distance from~$\FamML$ is computationally hard for each of our three distance measures. 
In particular, Theorem~\ref{thm:MLS-MLED-NPhard} presents a reduction from \FAS{} to MLS and also to MLED,
and Theorem~\ref{thm:MLVD-NPhard} provides a reduction from the \myproblem{Hitting Set} problem to MLVD.
Although both of these classic problems are $\NP$-hard, their approximability and 
their parameterized complexity for the standard parameterization---where the parameter is the desired value for the objective function---differs considerably, suggesting that the problem of computing the three distance 
measures may behave differently when viewed from the perspective of approximation or of parameterized complexity. 
Indeed, intrinsic differences between our three problems under examination will surface
when we consider their computational complexity in finer detail.

We start with Theorem~\ref{thm:MLS-MLED-NPhard} showing that we cannot expect a polynomial-time algorithm for MLS or for MLED
and even a polynomial-time approximation is unlikely to exist already for bipartite graphs.
As already mentioned, 
the proof of Theorem~\ref{thm:MLS-MLED-NPhard} relies on a connection between MLS, MLED, and the \FAS{} problem
which, given a directed graph~$D$ and an integer~$k$, asks whether there exists a set of at most~$k$ arcs in~$D$ whose deletion from~$D$
yields an acyclic graph.
Interestingly, the connection of this problem to MLS and to MLED can be used both ways: 
on the one hand, it serves as the basis of our reduction for proving computational hardness,
and on the other hand, we will be able to apply already existing algorithms for \FAS{} in our quest for solving MLS and MLED.

\begin{theorem}
\label{thm:MLS-MLED-NPhard}
MLS and MLED are both $\NP$-hard, 
and assuming the Unique Games Conjecture
 they are $\NP$-hard to approximate by any constant factor in polynomial time.
All of these hold even if the input graph is bipartite with all vertices on one side having degree~2, and preferences are strict.
\end{theorem}
 \begin{proof}
We give a reduction from the $\NP$-hard \myproblem{Feedback Arc Set} problem that proves the theorem for both MLS and MLED. 
Let digraph~$D=(V,E)$ and an integer~$k$ be the input of \myproblem{Feedback Arc Set}.

Given $D$, we create the following bipartite preference system~$\I=(G,\prec)$.
For each arc $e=(a,b) \in E$, we create a dummy vertex~$z_e$ whose neighbors are~$a$ and $b$, and we set $a \prec_{z_e} b$. 
We set $V \cup Z$ as the set of vertices in~$G$ where $Z= \{z_e:e \in E\}$.
To define the preferences of agents in~$V$, 
we fix an arbitrary strict ordering~$\prec^Z$ over the set $Z$ of dummy vertices,
and let $\prec_v$ be the restriction of~$\prec^Z$ to $N_G(v)$ for each $v \in V$.
Note that $G$ can be obtained from the undirected version of~$D$ by subdividing every edge, and therefore $G$ is bipartite. 
Moreover, the preference digraph $D_\I$ of~$\I$ is the disjoint union of $D_\I[Z]$, 
which is acyclic, and $D_\I[V]$, which is exactly the digraph~$D$, with each arc~$(a,b)$ labelled by~$z_{(a,b)}$.

We claim that the following three statements are equivalent: 
\begin{itemize}[topsep=4pt]
\item[(1)] $\distswap(\I) \leq k$;
\item[(2)] $\distedge(\I) \leq k$;
\item[(3)] there exists a feedback arc set of size at most~$k$ in~$D$.
\end{itemize}

First note that (1) $\Rightarrow$ (2) is implied by Lemma~\ref{prop:comparing-modulators}.

To prove (2) $\Rightarrow$ (3), 
suppose that $S$ is a set  of at most~$k$ edges in~$G$ such that~$\I-S$ 
admits a master list~$\prec^{\textup{ml}}$. 
Recall that $|N_G(z)|=2$ for each $z \in Z$ and $E(G) \subseteq Z \times V$. 
Define an arc set $F_S= \{e \in E : z_e  \textrm{ is incident to some edge in } S\}$; notice that $|F_S| \leq |S| \leq k$. 
Then $F$ is a feedback arc set, since the restriction of~$\prec^{\textup{ml}}$ on~$V$ is a topological order for~$V$ in~$D-F_S$.
Indeed, if $e=(a,b) \in E \setminus F_S$ then $(z_e,a),(z_e,b)$ are both present in $\I-S$, implying that $a \prec_{z_e} b$ 
is consistent with the master list in $\I-S$, that is, $a \prec^{\textup{ml}} b$.

To prove (3) $\Rightarrow$ (1), suppose that $F$ is a feedback arc set of size at most~$k$ in~$D$. 
W.l.o.g. we may assume that $F$ is inclusion-minimal, and thus reversing each arc of $F$ in $D$ results in an acyclic digraph~$D_{\overleftarrow{F}}$.
Accordingly, let us ``reverse'' the preferences of the agents corresponding to edges in~$F$,
that is, let $S$ be the set of swaps that switches $a$ and $b$ for each $e=(a,b) \in F$ in the preferences of~$z_e$. 
Note that swapping $a$ and $b$ in the preferences if~$z_{\{a,b\}}$ is admissible, since $N_G(z_e)=\{a,b\}$, 
and it is equivalent to reversing the arc $(a,b)$ in the preference digraph.
Thus we obtain that the preference digraph of the instance $\I \lhd S$ is acyclic, as 
it is the disjoint union of $D_\I[Z]$ and the digraph~$D_{\overleftarrow{F}}$. 
This means that $\I \lhd S \in \FamML$ by Observation~\ref{obs:pref-digraph},
implying $\distswap(\I) \leq |S| \leq k$. 

Hence, our claim holds and the reduction is correct, showing $\NP$-hardness for both MLS and MLED.

To show our inapproximability result, let $f(D)$ denote the instance~$\I$ created in the above reduction 
($\I$ is an instance of MLS as well as one of MLED). 
The above arguments prove that an optimal solution for $f(D)=\I$ has the same size as a feedback arc set for~$D$ of minimum size. 
Moreover, for any set~$S$ of edges in~$G$ that is a solution for our MLED instance~$\I$, 
the arc set~$F_S$ is a feedback arc set in~$D$ with $|F_S|\leq |S|$, 
so setting $g(S)=F_S$ yields that the pair $(f,g)$ is an approximation preserving reduction from \FAS{} to MLED.
Similarly, for any set~$S'$ of swaps in~$\I$ that is a solution for our MLS instance~$\I$, 
using the arguments of Lemma~\ref{prop:comparing-modulators} we can compute a set~$h(S')$ of at most~$|S'|$ edges in~$G$ for which $\I-h(S') \in \FamML$, 
so setting $g'(S')=F_{h(S')}$ yields that $(f,g')$ is an approximation preserving reduction from \FAS{} to MLS.
By the results of Guruswami~\cite{GMR08}, we know that no polynomial-time algorithm can obtain a constant-factor approximation
for \FAS{} unless the Unique Games Conjecture fails, 
ruling out a polynomial-time constant-factor approximation algorithm both for MLS and for MLED under the same assumption.
\qed
\end{proof}

Thanks to Lemma~\ref{lem:MLS-strict=FVS} below, for any strict preference system~$\I$ we can decide whether $\distswap(\I) \leq k$ 
for some $k \in \mathbb{N}$ by applying the FPT algorithm of Lokshtanov et al.~\cite{LRS18} 
for \FAS{} on the preference digraph~$D_\I$ and parameter~$k$.
Their algorithm runs in time $O(k!4^k k^6(n+m))$ on an input graph with~$n$ vertices and~$m$ arcs~\cite{LRS18}. 
If $G=(V,E)$ is the graph underlying~$\I$, then~$D_\I$ has~$|V|$ vertices and $O(|V|\! \cdot \! |E|)$ arcs, 
implying a running time of~$O(k! 4^k k^6 |V|\! \cdot \!|E|)$.

\begin{lemma}
\label{lem:MLS-strict=FVS}
For a strict preference system $\I$, $\distswap(\I) \leq k$ if and only if 
the preference digraph of~$\I$ admits a feedback arc set of size at most~$k$.
\end{lemma}
\begin{proof}
Let $\I=(G,\prec)$ be a strict preference system with~$G=(V,E)$ as its underlying graph, 
and let $D_{\I}$ be the preference digraph of~$\I$. 
Observe that performing a swap that is admissible in the preferences of some vertex~$v \in V$ corresponds to
reversing an arc in $D_{\I}$ with label~$v$.
Hence, a set~$S$ of swaps such that $\I \lhd S \in \FamML$ directly yields a feedback arc set of size~$|S|$ for the preference digraph~$D_{\I}$.

To see the other direction, we need to show that a feedback arc set~$F$ for~$D_{\I}$ 
can be turned into a set~$S$ of admissible swaps in~$\I$ with $|S| \leq |F|$. 
We can assume w.l.o.g. that $F$ is inclusion-minimal.  
Given~$F$, we iteratively find an arc~$f=(a,b) \in F$ with label~$v$ such that $(a,b)$ is an admissible swap in $v$'s preferences,
and then proceed with the instance~$\I'$ resulting from~$\I \lhd (a,b;v)$ and the set~$F \setminus \{f\}$
which is a feedback arc set for the preference digraph of~$\I \lhd (a,b;v)$.

It remains to show that we can always find a suitable arc $f=(a,b) \in F$ with label~$v$ such that 
$(a,b)$ is an admissible swap in $v$'s preferences. Suppose otherwise, and let $(a,b) \in F$ with label~$v$ be such that 
the distance between $a$ and $b$ in $v$'s preferences is as small as possible. 
As $(a,b;v)$ is not a suitable swap in~$\I$, there exists a vertex $c$ between $a$ and $b$ in $v$'s preferences.
This yields a path $P=(a,c,b)$ in $D_\I$ with both of its arcs having label~$v$. 
Since $F$ is inclusion-minimal, $P$ cannot be a path in $D_\I-F$, and so at least one of its arcs is in~$F$, 
contradicting our choice of~$f$. 
\qed
\end{proof}

\begin{mycorollary}
\label{cor:MLS-FPT}
If preferences are strict, then MLS is fixed-parameter tractable with parameter~$k$, 
and can be solved in time $O(k! 4^k k^6 |V|\!\! \cdot \!\! |E|)$.
\end{mycorollary}

We remark that 
the connection between \myproblem{Kemeny Score} and \FAS{} has already been noted by Bartholdi~III et al.~\cite{BTT89} who proved its $\NP$-completeness.
The variant of \myproblem{Kemeny Score} with incomplete votes was investigated by Betzler et al.~\cite{BFGNR09};
they proved its fixed-parameter tractability when parameterized by the desired value of the Kemeny score. 
This already implies that MLS is fixed-parameter tractable with parameter~$k$, 
although the running time we obtain in 
Corollary~\ref{cor:MLS-FPT} is better than the one stated in~\cite[Theorem~10]{BFGNR09}, 
mainly due to the improvement by Lokshtanov et al.~\cite{LMVGRM07} over the original FPT algorithm for \FAS{} by Chen et al.~\cite{CLLOR08}.

\medskip
In contrast to MLS, the MLED problem is $\mathsf{W}[1]$-hard with~$k$ as the parameter;
the reduction is from~\textsc{Multicolored Clique}~\cite{FHRV09}. 

\begin{theorem}
\label{thm:MLED-W1hard}
MLED is $\mathsf{W}[1]$-hard with parameter~$k$, even for strict preferences.
\end{theorem}
\begin{proof}
We are going to present a reduction from the $\mathsf{W}[1]$-hard \myproblem{Multicolored Clique} problem~\cite{FHRV09}, 
where the input is an undirected graph~$G$ and a parameter~$k$ with the vertex set of~$G$ partitioned into~$k$ 
independent sets $V^1,\dots,V^k$, and the task is to find a clique of size~$k$ in~$G$. 
Clearly, such a clique must contain one vertex from each set~$V^i$, $i \in [k]$. 
Let us denote by $v^i_1,v^i_2, \dots, v^i_{|V^i|}$ the vertices in $V^i$.
We are going to construct a preference system~$\I=(G',\prec)$ and an integer~$k'$ such that 
$G$ contains a clique of size~$k$ if and only if there exists a set~$S$ of at most~$k'$ edges in $G'$ such that $\I-S \in \FamML$.

Instead of defining $G'$ directly, we are going to define a digraph~$D$ that will be the preference digraph of~$\I$.
Furthermore, we will explicitly label only a subset of the arcs in~$D$, and declare the remaining arcs as \emph{fixed}, 
meaning that they represent a set $k'+1$ parallel arcs, each labelled by distinct dummy vertices. 
First, we create a set $\{a_e,a'_e, b_e, b'_e:e \in E(G)\}$ and a set $\{v,v':v \in V(G)\}$ of vertices. 
Next, for each $i \in [k]$ and each $v \in V^i$ we create a gadget~$D_v$ as follows. 
For each $j \in [k] \setminus \{i\}$, let $E^j(v)$ denote the set of edges in $G$ 
that connect $v$ with some vertex in~$V^j$, and let $n_v^j=|E^j(v)|$.
Let us create a bijection $\sigma_v^j: [n_v^j] \rightarrow E^j(v)$. 
The gadget~$D_v$ will contain $v$ as a source and $v'$ as a sink, and we define $k-1$ arc-disjoint paths from~$v$ to~$v'$: 
for each $1 \leq j <i $ we add the path $$P_v^j=(v,a_{\sigma_v^j(1)},a'_{\sigma_v^j(1)},a_{\sigma_v^j(2)},a'_{\sigma_v^j(2)},\dots, 
a_{\sigma_v^j(n_v^j)},a'_{\sigma_v^j(n_v^j)},v'),$$
and for each $i<j \leq k $ we add the path $$P_v^j=(v,b_{\sigma_v^j(1)},b'_{\sigma_v^j(1)},b_{\sigma_v^j(2)},b'_{\sigma_v^j(2)},\dots, 
b_{\sigma_v^j(n_v^j)},b'_{\sigma_v^j(n_v^j)},v').$$
If $E^j(v)=\emptyset$, then the corresponding path from~$v$ to~$v'$ consist solely of the arc~$(v,v')$.
For each arc $(a_e,a'_e)$ appearing in~$D_v$ we add $b_e$ as its label, and conversely,
for each arc $(b_e,b'_e)$ appearing in~$D_v$ we add $a_e$ as its label; we declare all remaining arcs in~$D$ as fixed.

Next, for each $i \in [k]$, we chain the gadgets corresponding to vertices of~$V^i$ into a cycle by adding an 
arc~$((v^i_\ell)',v^i_{\ell+1})$ for each $\ell \in [|V^i|]$, where we define $v^i_{|V^i|+1}:=v^i_1$. 
We declare each of these arcs as fixed.
Let $D^i$ denote the directed graph created thus far on the vertices~$\{V(D_v):v \in V^i\}$;
note that  for any two indices~$i \neq j$, the digraphs  $D^i$ and $D^j$ are vertex-disjoint, and there are no arcs running between them. 
Observe also that every cycle in~$D^i$ traverses all vertices~$\{v,v':v \in V^i\}$, 
using some path $P_v^j$ for each $v \in V^i$ to get from~$v$ to~$v'$.
We let the union of $D^1, \dots, D^k$ be the subgraph of~$D$ induced by all non-dummy vertices. 

To finish the construction of~$D$, it remains to take care of dummy vertices (recall that by definition, $V(D)=V(G')$).
Let $Z$ be the set of dummy vertices introduced while creating the fixed arcs, and let $\prec^Z$ be an arbitrary ordering on~$Z$. 
Note that each dummy $z \in Z$ is used only once, when creating some fixed arc $e_z$ in~$D$ 
(meaning  that $z$ is the label of one of the $k'+1$ parallel arcs which together constitute the fixed arc~$e_z$), 
so $N_{G'}(z)$ is the set of endpoints of~$e_z$. 
We then define the preferences of all non-dummy vertices so that they prefer non-dummy vertices to dummies, 
and their preferences between dummies are derived from~$\prec^Z$.
Observe that our instance $\I=(G',\prec)$ and its preference digraph~$D$ are now well-defined. 
Finally, we set $k'=\binom{k}{2}$.

We claim that $G$ admits a clique of size~$k$ if and only if 
there exists a set~$S$ of at most~$k'$ edges in $G'$ such that $\I-S \in \FamML$.

First assume that the vertices $x_1,\dots,x_k$ form a clique in~$G$. 
We define a corresponding set $S=\{ \{a_e,b_e\} : e =\{x_i,x_j\} \textrm{ for some }1 \leq i<j\leq k \}$.
Notice that deleting an edge~$\{a_e,b_e\}$ from~$\I$ corresponds to deleting two arcs from the preference subgraph~$D$,
namely the arc~$(a_e,a'_e)$ labelled with~$b_e$ and the arc~$(b_e,b'_e)$ labelled with~$a_e$. Hence, the deletion of~$S$ 
removes an arc from each of the paths~$P_{x_i}^j$ with $1 \leq i<j \leq k$, 
which ensures for each~$i \in [k]$ that there is no path from~$x_i$ to~$x'_i$. 
Hence, the preference digraph of~$\I-S$ is indeed acyclic,
and thus $\I-S \in \FamML$ by Observation~\ref{obs:pref-digraph}.

Assume now that there exists a set $S \subseteq E(G')$ with $|S| \leq k'$ such that $\I-S \in \FamML$. 
Notice that fixed arcs in~$D$ are also present as arcs in the preference digraph $D_S$ of~$\I-S$, by our bound $|S| \leq k'$.
Also, for each $i \in [k]$ there must exist some $x_i \in V^i$ for which there is no path in~$D_S$ from~$x_i$ to~$x'_i$, as otherwise
these paths together with the fixed arcs between gadgets form a cycle. Inspecting the gadget $D_{x_i}$, 
it is clear that the deletion of~$S$ must remove at least one arc on the path~$P_{x_i}^j$ for each $j \in [k]\setminus \{i\}$,
which can be achieved by deleting some edge of~$G'$ incident to~$a_e$ or to~$b_e$ for some~$e \in E^j(x_i)$. 
This means that we have to remove at least $k-1$ arcs from each $D^i$, a total of at least~$k(k-1)=2k'$ arcs.
This can only be achieved by the deletion of~$|S| \leq k'$ edges, if the deletion of each edge in~$S$ from~$G'$ results 
in the removal of exactly two arcs (as no edge deletion can remove more than~$2$ arcs from the 
preference digraph labelled by non-dummy vertices),
 and hence cannot involve any vertices of the form~$a'_e$ or $b'_e$. 
Therefore, $S$ must contain a set of $\binom{k}{2}$ edges of the form~$\{a_e,b_e\}$ whose deletion from~$G'$ 
results in the removal of an arc from each of the paths~$P^j_{x_i}$ for each $1 \leq i<j \leq k$. 
This is only possible if~$S$ consists of the edge set of a clique consisting of the vertices~$x_1,\dots,x_k$.
\qed
\end{proof}

Although Theorem~\ref{thm:MLED-W1hard} provides strong evidence that there is no FPT algorithm for MLED with parameter~$k$, 
and by Theorem~\ref{thm:MLS-MLED-NPhard} we cannot hope for a polynomial-time approximation algorithm for MLED either, 
our next result shows that combining these two approaches yields a way to deal with the computational hardness of the problem.
Namely, Theorem~\ref{thm:MLED-2approx-FPT} provides a 2-approximation for MLED whose running time is FPT with parameter~$k$.
This result again relies heavily on the connection between MLED and \FAS{}.

\begin{theorem}
\label{thm:MLED-2approx-FPT}
There exists an algorithm that achieves a 2-approximation for MLED if preferences are strict, and runs in FPT time with parameter~$k$. 
\end{theorem}

\noindent
{\bf 2-approximation FPT algorithm for MLED  (strict preferences).}
Let the strict preference system~$\I=(G,\prec)$ with underlying graph $G=(V,E)$ and~$k \in \mathbb{N}$ be our input for MLED.
 See Algorithm~\ref{alg:MLED-2approx-strict} for a formal description.

First, we construct a directed graph~$H_\I$ by setting 
%
\begin{eqnarray*}
V(H_\I)&=&V \cup \{a_v^-, a_v^+: \{a,v\} \in E\}, \\
A(H_\I)&=& \{(a_c^+,b_c^-): a,b,c \in V, a \prec_c b \}  \cup \, \{(a_v^-,a), (a,a_v^+): v \in V, a \in N_G(v)\}. 
\end{eqnarray*} 

Next, we compute a minimum feedback arc set~$F$ in~$H_\I$ using the algorithm by Lokshtanov et al.~\cite{LRS18}. 
Notice that w.l.o.g. we may assume that $F$ only contains arcs incident to some vertex in~$V$, 
as we can replace any arc $(a_c^+,b_c^-)$ with the sole arc leaving~$b_c^-$, namely $(b_c^-,b)$, 
since all cycles containing~$(a_c^+,b_c^-)$ must also go through $(b_c^-,b)$.

Finally, we return the set~$S_F=\{\{a,v\} \in E: (a,a_v^+) \in F \textrm{ or } (a_v^-,a) \in F\}$.

\begin{algorithm}[t]
\caption{Obtaining a 2-approximation for MLED on input~$(\I,k)$ with strict preferences}
\label{alg:MLED-2approx-strict}
\begin{algorithmic}[1]
\State Construct the graph $H_\I$.
\State Let $F$ be a solution for \myproblem{Feedback Arc Set} on input~$(H_\I,k)$.
\State Ensure that each arc in $F$ is incident to some vertex in~$V$ by replacing all arcs of~$F$ entering some $a_v^-$ with $(a_v^-,a)$.
\State Return $S_F=\{\{a,v\} \in E: (a,a_v^+) \in F \textrm{ or }  (a_v^-,a)  \in F\}$.
\end{algorithmic}
\end{algorithm}

\begin{proof}[of Theorem~\ref{thm:MLED-2approx-FPT}]
We prove that Algorithm~\ref{alg:MLED-2approx-strict} gives a 2-approximation,
that is, 
if there exists a set~$S$ of at most~$k$ edges in~$G$ such that~$\I-S$ admits a master list~$\prec^\ml$, 
then Algorithm~\ref{alg:MLED-2approx-strict} returns a solution for~$\I$ of size at most~$2|S| \leq 2k$.

For a given instance $\I=(G,\prec)$, consider the directed graph $H_{\I}$ constructed by the algorithm. 
Observe first that contracting for each $v \in V$ the subgraph of~$H_\I$ 
induced by  $\{v\} \cup \{v_a^+,v_a^-: a \in N_G(v)\}$ in~$H_{\I}$ 
yields exactly the preference digraph of~$\I$ (without the arc-labelling).
Furthermore, $H_{\I}$ is acyclic if and only if the preference subgraph of~$\I$ is acyclic. 

Take any vertex~$a \in V$ incident to some edge in~$S$; let $\{a,b_1\}, \dots, \{a,b_t\}$ be the edges of~$S$ incident to~$a$.
Instead of deleting these edges, it is also possible to ``move'' the vertices in $B=\{b_1, \dots, b_t\}$ 
within the preferences of~$a$ to obtain an instance $\I'$ whose preferences are consistent with~$\prec^\ml$: 
when representing $a$'s preferences as a list, each $b \in B$ needs to be moved 
either to the left (becoming more preferred by~$a$) or to the right (becoming less preferred by~$a$). 
In the former case we mark the arc~$(b,b_a^+)$, and in the latter case we mark~$(b_a^-,b)$. 
Clearly, we have marked at most~$2|S|\leq 2k$ arcs in~$H_{\I}$, since for each edge~$\{a,b\} \in S$ 
we may need to move $a$ in $b$'s preferences, 
and also to move $b$ in $a$'s preferences, and each such move results in marking one arc. 

We claim that the set $M$ of marked arcs is a feedback arc set in~$H_{\I}$. 
Suppose for contradiction that there is a cycle $C$ in~$H_{\I}-M$. 
Observe that $C$ can be decomposed into paths of length~3 of the form~$(a,a_c^+,b_c^-,b)$. 
For such a path it follows that $a \prec^\ml_c b$ holds: first, $a \prec_c b$ holds in~$G$ by construction, 
and we have not moved $a$ up in~$c$'s preferences (by $(a,a_c^+) \notin M$), 
nor have we moved $b$ down in~$c$'s preferences (by $(b_c^-,b) \notin M$), so $c$ still prefers $b$ to~$a$ 
in the instance~$\I'$ whose preferences are consistent with~$\prec^\ml$. Therefore, the cycle~$C$ implies 
a cycle in the preference digraph of~$\I-S$, a contradiction. 
Thus, $M$ is a feedback arc set of size at most~$2k$. 

This implies that Algorithm~\ref{alg:MLED-2approx-strict} will produce a feedback arc set~$F$ for~$H_{\I}$ of size at most~$2k$ on line~2.
It remains to show that the set~$S_F$ is indeed a solution for our instance~$\I$ of MLED.
Consider the changes in the digraph~$H_\I$ as a result of deleting the edges of~$S_F$ from~$\I$.
Note that the deletion of~$S_F$ from~$\I$ ensures the removal of all arcs in~$F$ from~$H_\I$	
(since deleting an edge~$\{a,b\}$ from~$G$ corresponds to removing the 
vertices $a_b^+, a_b^-, b_a^+,$ and $b_a^-$ and all arcs incident to them from~$H_\I$),
so the resulting digraph~$H_{\I-S_F}$ is acyclic. 
Recall that $H_{\I-S_F}$ is acyclic if and only if the preference digraph of~$\I-S_F$ is acyclic. 
Thus by Observation~\ref{obs:pref-digraph} we get $\I-S_F \in \FamML$. 
This proves that Algorithm~\ref{alg:MLED-2approx-strict} yields a 2-approximation.

Note that $H_\I$ has $|V|+2|E|$ vertices and at most~$|V|\!\! \cdot \!\! |E|+4|E|$ arcs. 
The total running time of Algorithm~\ref{alg:MLED-2approx-strict} is therefore $O(k!4^k k^6|V|\!\! \cdot \!\! |E|)$ 
which is indeed FPT with parameter~$k$.
\qed
\end{proof}

Contrasting our positive results for MLS and MLED, a reduction from the classic~\textsc{Hitting Set} problem~\cite{karp-1972}
 shows that MLVD
is computationally hard both in the classic and in the parameterized sense, and cannot be 
approximated by any FPT algorithm, as stated by Theorem~\ref{thm:MLVD-NPhard}.

\begin{theorem}
\label{thm:MLVD-NPhard}
MLVD is $\NP$-hard and $\mathsf{W}[2]$-hard with parameter~$k$. 
Furthermore, no FPT algorithm with~$k$ as the parameter can achieve an $f(k)$-approximation for MLVD for any computable function~$f$, 
unless $\mathsf{FPT}=\mathsf{W}[1]$. 
All of these hold even if the input graph is bipartite and preferences are strict.
\end{theorem}
\begin{proof}
We present a reduction from the \HS{} problem. 
The input of \HS{} is a set~$U$ (the \emph{universe}), a family $S_1, \dots, S_m$ of subsets of~$U$, and an integer~$k$;
the task is to find a hitting set of size at most~$k$, where a \emph{hitting set} is a 
set~$H \subseteq U$ such that $H \cap S_i \neq \emptyset$ for each $i \in [m]$. 

Given~an instance~$\mathcal{H}=(U,S_1,\dots,S_m,k)$, let $u(S_i,j)$ be the $j$-th item in~$S_i$ for each $j \in [|S_i|]$. 
We construct a bipartite preference system $\I=(G,\prec)$ as follows. 
For each $i \in [m]$ and $j \in [|S_i|]$, we create an agent $x_{S_i}^j$. 
The vertex set of~$G$ is $V=X \cup U$ where $X=\{ x_{S_i}^j : i \in [m], j \in [|S_i|]\}$, and 
we define the edge set of~$G$ by connecting $u(S_i,j)$ for each $i \in [m]$ and $j \in [|S_i|]$ to both $x_{S_i}^j$ and~$x_{S_i}^{j+1}$ 
(where we define $x_{S_i}^{|S_i|+1}=x_{S_i}^1$ for each~$i$);
we set also $x_{S_i}^j \prec_{u(S_i,j)} x_{S_i}^{j+1}$. Note that $G$ is bipartite.
To define the full preferences of agents in~$U$, 
we further add that for any agent~$u \in U$ whose neighborhood contains some agents~$x_{S_i}^j$ and~$x_{S_{i'}}^{j'}$ with $i \neq i'$, 
we let $x_{S_i}^j \prec_u x_{S_{i'}}^{j'}$ exactly if $i<i'$. 
To finish the definition of the preferences in~$\I$, we fix an arbitrary
strict ordering $\prec^U$ over the universe~$U$ as a master list over~$U$, 
i.e., for any $x \in X$ we set $u \prec_x u'$ for some $u$ and~$u'$ in~$U$ if and only if~$u \prec^U u'$.

We claim that there exists a hitting set of size at most~$k$ for $\mathcal{H}$ if and only if $\distvert(\I)\leq k$.

Suppose first that $S$ is a set of at most~$k$ 
vertices in~$G$ such that $\I-S$ admits a master list~$\prec^{\textup{ml}}$. 
Let $H_S$ be the union of $S \cap U$ and $\{u(S_i,j): x_{S_i}^j \in X \cap S\}$. 
Note that~$|H_S| \leq |S| \leq k$; we show that~$H_S$ is a hitting set for~$\mathcal{H}$. 
To see this, suppose that a set $S_i$ for some $i \in [m]$ is not hit by~$H_S$, that is, $S_i \cap H_S=\emptyset$. 
Then by definition, $S_i \cup \{x_{S_i}^j:j \in [|S_i|]\}$ is disjoint from~$S$. 
We show that the sub-instance of $\I-S$ consisting of these $2|S_i|$ vertices cannot be consistent with the master list~$\prec^{\textup{ml}}$:
if there exists some $\ell \in [|S_i|-1]$ for which $x_{S_i}^\ell \not\prec^{\textup{ml}} x_{S_i}^{\ell+1}$, 
then this contradicts $x_{S_i}^\ell \prec_{u(S_i,\ell)} x_{S_i}^{\ell+1}$;
otherwise we know $x_{S_i}^1 \prec^{\textup{ml}} x_{S_i}^2  \prec^{\textup{ml}} \dots  \prec^{\textup{ml}} x_{S_i}^{|S_i|}$ which 
contradicts $x_{S_i}^{|S_i|}  \prec_{u(S_i,|S_i|)} x_{S_i}^{1}$. This proves that $H_S$ is indeed a hitting set for $\mathcal{H}$.

Suppose now that $H$ is a hitting set for $\mathcal{H}$; 
we will show that $H$ itself is a solution for our instance~$\I$ of MLVD by defining a master list~$\prec^\ml$ for $\I-H$.
For each $i \in [m]$, let~$\ell_i$ be the index for which $u(S_i,\ell_i) \in H$; such an index~$\ell_i$ exists
by the definition of a hitting set. 
Then we set $$x_{S_i}^{\ell_i+1} \prec^\ml x_{S_i}^{\ell+2} \prec^\ml \dots 
\prec^\ml x_{S_i}^{|S_i|} \prec^\ml x_{S_i}^{1} \prec^\ml \dots \prec^\ml x_{S_i}^{\ell_i}.$$  
Furthermore, we set $x_{|S_i|}^{\ell_i} \prec^\ml x_{S_{i+1}}^{\ell_{i+1}+1}$ for each $i \in [m-1]$, 
so that $\prec^\ml$ is now a strict linear order over $X$. 
It is clear that the preferences of each agent $u \in U \setminus H$ in~$\I-H$  are consistent with $\prec^\ml$.
Finally, we let $\prec^\ml$ order~$U$ according to $\prec^U$; then preferences of agents in~$U$ are consistent with $\prec^\ml$ as well, 
proving that $\I-H \in \FamML$, which in turn proves the correctness of the reduction.

Observe that the above reduction is a polynomial-time reduction as well as an FPT-reduction, 
if the parameter is~$k$ both in the instance~$\mathcal{H}$ of \HS{} and in MLVD. 
Since \HS{} is $\NP$-hard and $\mathsf{W}[2]$-complete with parameter~$k$ (see~\cite{karp-1972},\cite{downey-fellows-FPC-book}), 
this shows that MLVD is $\NP$-hard and $\mathsf{W}[2]$-hard with parameter~$k$. 
To show our inapproximability result, let $f(\mathcal{H})$ denote the instance~$\I$ created in the above reduction.
The above arguments prove that the size of an optimal solution for $f(\mathcal{H})=\I$ equals the size of a 
hitting set for~$\mathcal{H}$ of minimum size. 
Moreover, for any set~$S$ of vertices in~$G$ that is a solution for our MLVD instance~$\I$, 
the set~$H_S$ is a hitting set in~$\mathcal{H}$ with $|H_S|\leq |S|$, 
so setting $g(S)=H_S$ yields that the pair $(f,g)$ is an approximation preserving reduction from \HS{} to MLVD.
By the results of Karthik et al.~\cite{KLM2019}, we know that no FPT algorithm with parameter~$k$ can obtain an $f(k)$-approximation
for \HS{} for any computable function~$f$, unless $\mathsf{FPT}=\mathsf{W}[1]$.\footnote{In fact, 
Karthik et al.~\cite{KLM2019} deal with the \myproblem{Dominating Set} problem, 
which is a subset of the \myproblem{Set Cover} problem, which is in turn well known to be equivalent with \HS{}. 
Thus, the results in~\cite{KLM2019} carry over for \HS{}.}
Hence, our reduction implies the inapproximability statement of the theorem.
\qed
\end{proof}

\subsection{Weakly ordered preferences}
\label{sec:main-res-ties}

Let us now consider preference systems that are not necessarily strict.
The hardness results of Section~\ref{sec:main-res-strict} trivially hold for weakly ordered preferences,
so we will focus on extending the algorithmic results of the previous section. 
We start by showing that Lemma~\ref{lem:MLS-strict=FVS} can be generalized for instances with weakly ordered preferences.

\begin{lemma}
\label{lem:MLS-ties=DFAS}
For any preference system~$\I=(G,\preceq)$, $\distswap(\I) \leq k$ if and only if
there exists a set of at most~$k$ arcs in the preference digraph~$D_\I$ of~$\I$ that hits every strict cycle of~$D_\I$.
\end{lemma}
\begin{proof}
Suppose first that $\distswap(\I) \leq k$. 
This means that there exists a preference system~$\I'=(G,\preceq') \in \FamML$ such that $\Delta(\I,\I') \leq k$.
Let $V$ denote the vertex set of~$G$, and let
$$S_v=\{ \{a,b\}:a \prec_v b \textrm{ but } b \preceq'_v a, \textrm{ or } a \sim_v b \textrm{ but } a \not\sim'_v b\}.$$
By definition, $$\Delta (\I,\I')=\sum_{v \in V} \Delta(\preceq_v,\preceq'_v)= \sum_{v \in V}|S_v| \leq k.$$

We construct an arc set~$F$ in~$D_\I$ by adding an arc to~$F$ for each $\{a,b\} \in S_v$ for some~$v$ as follows:
\begin{itemize}
\item[(i)] if $a \prec_v b$ but $b \preceq'_v a$, then we add the strict arc~$(a,b)$ with label~$v$ to~$F$;
\item[(ii)] if $a \sim_v b$ but $a \prec'_v b$, then we add the tied arc~$(b,a)$ with label~$v$ to~$F$. 
\end{itemize}
Note that $|F| \leq \sum_{v \in V}|S_v| \leq k$.

We claim that deleting~$F$ from~$D_\I$ yields a subgraph of~$D_{\I'}$, and additionally,
all arcs that are strict in~$D_\I$ and are not in~$F$ are also strict in $D_{\I'}$.
To see this, let us compare $D_\I$ and $D_{\I'}$. Note that both $\I$ and $\I'$ have $G$ as their underlying graph.
Therefore, it suffices to consider the changes in the preference digraph when preferences change from~$\preceq$ to~$\preceq'$, 
which amounts to considering each pair~$\{a,b\} \in \bigcup_{v \in V} S_v$. 
Let us fix some~$v$, and consider $\{a,b\} \in S_v$. 
First, if~$a \prec_v b$, then $D_\I$ contains a strict arc $(a,b)$ with label~$v$, 
but this arc has been added to~$F$ in case~(i), and hence there is no arc with label~$v$ with endpoints~$a$ and $b$ in~$D_\I-F$.
Second, if $a \sim_v b$ but $a \prec'_v b$, then $D_\I$ contains a pair of tied arcs with label~$v$ between $a$ and~$b$, 
but the one pointing from~$b$ to~$a$ has been added to~$F$ in case~(ii), and hence only the arc~$(a,b)$ with label~$v$ remains in~$D_\I-F$;
by $a \prec' b$, such an arc is also present in~$D_{\I'}$ as a strict arc.
This proves our claim. 

As a consequence, any cycle~$C$ in~$D_\I$ that contains a strict arc and is not hit by~$F$ (i.e., $A(C) \cap F=\emptyset$) 
is also present in~$D_{\I'}$, and contains an arc that is strict in~$D_{\I'}$. In other words, 
any strict cycle in~$D_\I$ corresponds to a strict cycle in~$S_{\I'}$.
However, by Lemma~\ref{lem:pref-digraph-ties} no strict cycle can be present in~$D_{\I'}$, as $\I' \in \FamML$. 
This shows that $F$ hits all strict cycles in~$D_\I$.

To prove the other direction of our statement, suppose now that~$F$ is a set of at most~$k$ arcs in~$D_\I$ that 
hits all strict cycles. 
Let us call a path in~$D_\I$ a \emph{strict path} if it contains a strict arc.
We proceed with the following iteration on~$D_\I-F$: 
as long as there exist distinct vertices $a$ and $b$ such that there is no strict path either from~$a$ to~$b$ or from~$b$ to~$a$, 
add a pair of tied arcs between $a$ and $b$. Let $D'$ be the digraph obtained at the end of this iteration. 
Observe that~$D'$ still has the property that it contains no cycle passing through a strict arc. 

Let us now define a master list~$\prec^\ml$ by setting $a \prec^\ml b$ whenever there is a strict path from~$a$ to~$b$. 
Observe that the relation~$\prec^\ml$ is irreflexive (as a strict path from~$a$ to~$a$ would yield a strict cycle), 
asymmetric (as a strict path from~$a$ to~$b$ and from~$b$ to~$a$ implies the existence of a strict cycle),
and transitive as well (as a strict path from~$a$ to~$b$ together with one from~$b$ to~$c$ 
yields a strict path from~$a$ to~$c$, due to the absence of strict cycles). 
Thus, $\prec^\ml$ is a partial order, and it remains to show that indifference w.r.t.~$\prec^\ml$, that is, $\sim^\ml$ is transitive. 
And indeed, if $a \sim^\ml b$ and $b \sim^\ml c$ for vertices $a,b,c \in V$, then as a results of our iterative process 
there must exist a pair of tied arcs between $a$ and $b$, and also between $b$ and $c$. 
Therefore, $a \prec^\ml c$ is not possible, as a strict path from~$a$ to~$c$ together with the arcs $(c,b)$ and $(b,a)$ would
imply the existence of a strict cycle in~$D'$, a contradiction. 

Let $\I'$ be the preference system whose underlying graph is~$G$ and preferences are derived from~$\preceq^\ml$; then $\I' \in \FamML$.
We claim $\Delta(\I,\I') \leq |F| \leq k$, which would imply $\distswap(\I) \leq k$, and hence would finish the proof. 
For each $v \in V$ let us define $$S_v=\{ \{a,b\}: a \prec_v b \textrm{ but } b \preceq^\ml a, \textrm{ or } a \sim_v b \textrm{ but } a \not\sim^\ml b\}.$$
Consider any pair~$\{a,b\} \in S_v$ for some $v \in V$. 
First, if $a \prec_v b$ but $a \not\prec^\ml b$, then the strict arc~$(a,b)$ with label~$v$ must be contained in~$F$, as otherwise 
it would imply a strict path from~$a$ to~$b$ in~$D'$, contradicting  $a \not\prec^\ml b$. 
Second, if $a \sim_v b$ but $a \prec^\ml b$, then the arc~$(b,a)$ (part of a pair of tied arcs in~$D_\I$) must be contained in~$F$,
as otherwise together with the strict path from~$a$ to~$b$ (whose existence follows from $a \prec^\ml b$) 
it would yield a cycle passing through a strict arc in~$D'$. 
Therefore we know $\sum_{v \in V} |S_v| \leq |F| \leq k$, proving our claim.
\qed
\end{proof}

Thanks to Lemma~\ref{lem:MLS-ties=DFAS}, we can reduce MLS to a generalization of the \FAS{} problem where, 
instead of searching for a feedback arc set, 
the task is to seek an arc set that only hits certain \emph{relevant} cycles.
In the \myproblem{Subset Feedback Arc Set} (or SFAS) problem the input is 
a directed graph~$D$, a vertex set~$W \subseteq V(D)$ and an integer~$k$, 
and the task is to find a set of at most~$k$ \emph{arcs} in~$D$ that hits  all \emph{relevant} cycles in~$D$,  
where a cycle is relevant if it goes through some vertex of~$W$. 

To solve SFAS, we apply an FPT algorithm by Chitnis et al.~\cite{CCHM15} for the vertex variant of SFAS, 
 the \myproblem{Directed Subset Feedback Vertex Set} (or DSFVS) problem that, 
given a directed graph~$D$, a set~$W \subseteq V(D)$ and a parameter~$k \in \mathbb{N}$, 
asks for a set of at most~$k$ \emph{vertices} that hits all relevant cycles in~$D$.   
Applying a simple, well-known reduction from SFAS to DSFVS, we can use the algorithm by Chitnis et al.~\cite{CCHM15} 
to obtain an FPT algorithm for MLS with parameter~$k$.

\begin{theorem}
\label{cor:MLS-ties-FPT}
MLS is fixed-parameter tractable with parameter~$k$, even if preferences are weak orders.
\end{theorem}
\begin{proof}
Due to Lemma~\ref{lem:MLS-ties=DFAS}, we know that we can reduce MLS to an instance of the \myproblem{Subset Feedback Arc Set} problem.
To reduce SFAS to~DSFVS (that is, \myproblem{Directed Subset Feedback Vertex Set}), 
we will apply a well-known reduction that is often used to reduce arc-deletion problems 
to vertex-deletion problems: given our input instance~$(D,W,k)$ for SFAS 
we first subdivide each arc~$e \in A(D)$ by a new vertex~$a_e$, 
and then replace each vertex~$v \in V(D)$ by a set~$X_v=\{v^1,\dots, v^{k+1}\}$ of vertices containing $k+1$ copies of~$v$,
with each~$v^i$ having the same in- and out-neighbors as~$v$ (after the subdivisions).
Let $D'$ denote the directed graph we obtain. 
It is not hard to see that $(D,W,k)$ is a `yes'-instance of SFAS if and only if $(D',W',k)$ is a `yes'-instance of DSFVS 
where $W'=\bigcup\{X_v:v \in W\}$, as it makes no sense to delete any vertex in $X_v$ for some $v \in V(D)$
(since at least one of the $k+1$ copies of~$v$ in~$X_v$ will survive), 
and deleting a vertex~$a_e$ in~$D'$ is equivalent with deleting the arc~$e$ in~$D$.
Thus, by the FPT algorithm of Chitnis et al.~\cite{CCHM15} for DSFVS, we know that SFAS is also FPT with parameter~$k$.
\qed
\end{proof}

Next we extend Theorem~\ref{thm:MLED-2approx-FPT}  for weak orders, by reducing MLED to SFAS.
\begin{theorem}
\label{thm:MLED-ties-2approx-FPT}
There exists an algorithm that achieves a 2-approximation for MLED, 
and runs in FPT time with parameter~$k$. 
\end{theorem}

\noindent
{\bf 2-approximation FPT algorithm for MLED.}
Let the preference system~$\I$ with underlying graph $G=(V,E)$ and $k \in \mathbb{N}$ be our input for MLED.
For each vertex~$v \in V$, let~$T_v$ be the set family containing every tie that appears in the preferences of~$v$. 
See Algorithm~\ref{algo:MLED-2approx} for a formal description. 

First, we construct a directed graph $H_\I$ with $V(H_\I)=V \cup T \cup U \cup Z$
where 
\begin{eqnarray*}
T &=& \{t:v \in V, t \in T_v\}, \\
U &=& \{a_v^-, a_v^+: \{a,v\} \in E\}, \\
Z &=& \{z_{(a,b,v)}: a \prec_v b \textrm{ for some $a,b,v \in V$} \}, 
\end{eqnarray*}
and with arc set $A(H_\I)=A_T \cup A_U \cup A_Z$ where 
\begin{eqnarray*}
A_T &=& \{(t,a_v^-),(a_v^+,t):v \in V, t \in T_v ,a \in t\} \\
A_U &=& \{(a_v^-,a),(a,a_v^+):v \in V, a \in N_G(v)\} \\
A_Z &=& \{(a_v^+,z_{(a,b,v)}), (z_{(a,b,v)},b_v^-): z_{(a,b,v)} \in Z \}. 
\end{eqnarray*}

Next, we solve the \myproblem{Subset Feedback Arc Set} problem~$(H_\I,Z,k)$ by applying the above reduction from SFAS to DSFVS 
and then using the algorithm of Chitnis et al.~\cite{CCHM15};
let~$F$ be the solution obtained for $(H_\I,Z,k)$. 
Observe that  we may assume that $F$ only contains arcs of~$A_U$. 
Indeed we can replace any arc~$f \in F$ in~$A_T \cup A_Z$ by an appropriately chosen arc~$f' \in A_U$:
note that $f$ either points to some $a_v^- \in U$ or it leaves some $a_v^+ \in U$;
in the former case we set~$f'=(a_v^-,a)$, while in the latter case we set $f'=(a,a_v^+)$. 
Then any cycle containing~$f$ must also contain~$f'$, so we can safely replace~$f$ with~$f'$, 
as~$F \setminus \{f\} \cup \{f'\}$ still hits all relevant cycles. 
Hence, we will assume $F \subseteq A_U$.

Finally, we return the set $S_F=\{\{a,v\} \in E: (a,a_v^+) \in F \textrm{ or } (a,a_v^+) \in F\}$.

\begin{algorithm}[t]
\caption{Obtaining a 2-approximation for MLED on input~$(\I,k)$}
\label{algo:MLED-2approx}
\begin{algorithmic}[1]
\State Construct the graph $H_\I$.
\State Let $F$ be a solution for \myproblem{Subset Feedback Arc Set} on input~$(H_\I,Z,k)$.
\State Ensure $F \subseteq A_U$ by replacing all arcs pointing to some $a_v^- \in U$ with $(a_v^-,a)$ 
and all arcs leaving some $a_v^+ \in U$ with $(a,a_v^+)$.
\State Return $S_F=\{\{a,v\} \in E: (a,a_v^+) \in F \textrm{ or } (a_v^-,a) \in F\}$.
\end{algorithmic}
\end{algorithm}

\begin{proof}[of Theorem~\ref{thm:MLED-ties-2approx-FPT}]
It is clear that Algorithm~\ref{algo:MLED-2approx} runs in FPT time with parameter~$k$, so it suffices to show that it yields a 2-approximation.
Let us show that if there exists a set~$S$ of at most~$k$ edges in~$G$ such that~$\I-S$ admits a master list~$\prec^\ml$, 
then Algorithm~\ref{algo:MLED-2approx} returns a solution for~$\I$ of size at most~$2|S| \leq 2k$.

Take any vertex~$a \in V$ incident to some edge in~$S$; let $\{a,b_1\}, \dots, \{a,b_t\}$ be the edges of~$S$ incident to~$a$.
Instead of deleting these edges, it is also possible to ``move'' the vertices in $B=\{b_1, \dots, b_t\}$ 
within the preferences of~$a$ 
to obtain an instance $\I'$ whose preferences are consistent with~$\prec^\ml$: 
when representing $a$'s preferences as a list possibly containing ties, each $b \in B$ needs to be moved 
either to the left (becoming more preferred by~$a$) or to the right (becoming less preferred by~$a$). 
In the former case we mark the arc~$(b,b_a^+)$, and in the latter case we mark~$(b_a^-,b)$. 
Clearly, we have marked at most~$2|S|\leq 2k$ arcs in~$H_\I$, since for each edge~$\{a,b\} \in S$ 
we may need to move $a$ in $b$'s preferences, 
and also to move $b$ in $a$'s preferences, and each such move results in marking one arc. 

We claim that the set $M$ of marked arcs hits every relevant cycle in~$H_\I$, that is, 
$H_{\I}-M$ contains no cycle passing through a vertex of~$Z$. 
Suppose for contradiction that there is a cycle $C$ in~$D_{\I}-M$ going through a vertex of~$Z$. 
Let us consider the decomposition of~$C$ into paths $P_1, \dots, P_r$ such that $V \cap V(P_i)$ 
contains exactly the endpoints of $P_i$ for each $i \in [r]$. 
Let $P_i$ be such a path from~$a$ to~$b$. Observe that there are two possible forms $P_i$ can take:
either 
\begin{itemize}
\item[(i)] $P_i=(a,a^+_v,t,b^-_v,b)$ for some $v \in V$ and $t \in T_v$, or
\item[(ii)] $P_i=(a,a^+_v,z_{(a,b,v)},b^-_v,b)$. 
\end{itemize}
In both cases it follows that $a \preceq^\ml b$ holds: 
first, $a \preceq_v b$ holds in~$G$ by construction (since $a \sim_v b$ in case (i) and $a \prec_v b$ in case (ii)), 
and we have not moved $a$ up in~$v$'s preferences (because $(a,a_v^+) \notin M$), 
nor have we moved $b$ down in~$v$'s preferences (because $(b_v^-,b) \notin M$), so $v$ still weakly prefers $b$ to~$a$ 
in the instance~$\I'$ whose preferences are consistent with~$\prec^\ml$. 
Furthermore, in case (ii) we know $a \prec_v b$, and using the same reasoning we obtain $a \prec^\ml b$ as well.
Therefore, the cycle~$C$ implies a cycle in the preference digraph~$D_{\I'}$ of~$\I'$, 
and since $C$ contains a vertex of~$Z$, at least one arc of this cycle is a strict arc
(namely, if $z_{(a,b,v)} \in V(C)$, then the arc $(a,b)$ with label~$v$ is a strict arc in $D_{\I'}$), 
which contradicts~$\I' \in \FamML$ by Lemma~\ref{lem:pref-digraph-ties}.
Thus, $M$ hits every relevant cycle in $H_\I$, and has size at most~$2k$.

This implies that Algorithm~\ref{algo:MLED-2approx} will produce an arc set~$F$ for~$H_{\I}$ of size at most~$2k$ on line~2.
It remains to show that the set~$S_F$ is indeed a solution for our instance~$\I$ of MLED.
For this, we will need the observation that, due to the construction of~$H_\I$, 
we can obtain the preference digraph $D_\I$ from $H_\I$ by applying the following operations:
\begin{enumerate}
\item For each $v \in V$, contract all arcs running between~$v$ and $U$. 
Note that each such arc is either the unique outgoing or the unique incoming arc of some vertex of~$U$.
We will keep referring to the node resulting from these contractions as~$v$, corresponding directly to the vertex~$v$ in~$D_\I$. 
\item For each $t \in T_v$ for some $v \in V$, replace~$t$ and its neighborhood in~$H_\I$
(which induces a bidirected star with center~$t$) with a bidirected clique on $N_{H_\I}(t)$, i.e., 
for each $\{a,b\} \subseteq t$, add the pair $(a,b)$ and $(b,a)$ of arcs with label~$v$. 
These arcs correspond to a pair of tied arcs in~$D_\I$.  
\item Lift each $z_{(a,b,v)} \in Z$, i.e., delete $z_{(a,b,v)}$ and add an arc $(a,b)$ with label~$v$. 
These arcs correspond to strict arcs in~$D_\I$.
\end{enumerate}
It follows that there is a bijection between the cycles in~$H_\I$ passing through some vertex of~$Z$
and the cycles in~$D_\I$ passing through a strict arc.

Now, consider how the digraph~$H_\I$ changes as a result of deleting the edges of~$S_F$ from~$\I$.
Note that the deletion of~$S_F$ ensures the removal of all arcs in~$F$
(since deleting an edge~$\{a,b\}$ from~$G$ results in the removal of
vertices $a_b^+, a_b^-, b_a^+,$ and $b_a^-$ and all arcs incident to them from~$H_\I$),
so the resulting digraph~$H_{\I-S_F}$ contains no cycles passing through a vertex of~$Z$. 
As we proved in the previous paragraph, $H_{\I-S_F}$ contains no cycles passing through a vertex of~$Z$ 
if and only if the preference digraph of~$\I-S_F$ contains no cycles passing through a strict arc. 
Therefore by Lemma~\ref{lem:pref-digraph-ties} we get $\I-S_F \in \FamML$.
\qed
\end{proof}

\section{Applications}
\label{sec:applic}

In this section we consider two examples related to stable and popular matchings where 
we can efficiently solve computationally hard optimization problems 
on preference systems that are close to admitting a master list.

\subsection{Optimization over stable matchings}
\label{sec:enum-bounds}
One of the most appealing property of the distances defined in Section~\ref{sec:initial_results} is that 
whenever the distance of a \emph{strict} (but not necessarily bipartite) preference system from admitting a master list is small, 
we obtain an upper bound on the number of stable matchings contained in the given preference system. 
Therefore, strict preference systems that are close to admitting a master list 
are easy to handle, as we can efficiently enumerate their stable matchings, 
as Lemmas~\ref{lem:SM-enum-edge} and~\ref{lem:SM-enum-vert}~show.
We omit presenting a proof for these lemmas, because we will prove more generals versions of them in Section~\ref{sec:applic-pop},
namely, Lemmas~\ref{lem:enum-fixedblocks-edge} and~\ref{lem:enum-fixedblocks-vert}.

\begin{mylemma}
\label{lem:SM-enum-edge}
Given a strict preference system $\I=(G,\prec)$ with~$G=(V,E)$ 
and a set $S \subseteq E$ of edges such that $\I-S \in \FamML$, 
the number of stable matchings in~$\I$ is at most $2^{|S|}$, and 
it is possible to enumerate all of them in time $2^{|S|} \cdot O(|E|)$.  
\end{mylemma}

Observe that Corollary~\ref{lem:SM-enum-swap} follows immediately from Lemma~\ref{lem:SM-enum-edge} and the fact that 
$\distedge(\I) \leq \distswap(\I)$ for any preference system~$\I$ (as we know by Lemma~\ref{prop:comparing-modulators}). 

\begin{mycorollary}
\label{lem:SM-enum-swap}
In a strict preference system~$\I$, the number of stable matchings is at most~$2^{\distswap(\I)}$.
\end{mycorollary}

Notice that although the number of stable matchings may grow exponentially as a function of
the distance~$\distedge$ or $\distswap$, this growth does not depend on the size of the instance. 
By contrast, this is not the case for the distance~$\distvert$, as stated by Lemma~\ref{lem:SM-enum-vert} below.

\begin{mylemma}
\label{lem:SM-enum-vert}
Given a strict preference system $\I=(G,\prec)$ with~$G=(V,E)$ 
and a set $S \subseteq V$ of vertices such that $\I-S \in \FamML$, 
the number of stable matchings in~$\I$ is at most $|V|^{|S|}$, and 
it is possible to enumerate all stable matchings of~$\I$ in time $|V|^{|S|} \cdot O(|E|)$.  
\end{mylemma}

There exists an algorithm by Gusfield and Irving~\cite{Gusfield88,GusfieldIrving1989} that outputs the set~$\S(\I)$ of stable matchings 
in a preference system~$\I$ over a graph $G=(V,E)$ in~${O(|\mathcal{S}(\I)| \!\cdot\! |E|)}$ time
after $O(|V|\cdot|E| \log |V|)$ preprocessing time.
As a consequence, the bounds of Lemma~\ref{lem:SM-enum-edge}, Corollary~\ref{lem:SM-enum-swap}, and Lemma~\ref{lem:SM-enum-vert} 
on~$|\S(\I)|$ directly yield a way to handle computationally hard problems on 
any preference system~$\I$ where $\distswap(\I)$, $\distedge(\I)$, or $\distvert(\I)$ has small value,
even without the need to determine a set~$S$ of edges or vertices for which $\I-S \in \FamML$ 
or a set $S$ of swaps for which $I \lhd S \in \FamML$.
Thus, we immediately have the following result, even without having to use our results in Section~\ref{sec:main-res}.
For the definitions of the $\mathsf{NP}$-hard problems mentioned as an example in Theorem~\ref{thm:stability-FPT-XP},
see the book~\cite{Manlove2013}.

\begin{theorem}
\label{thm:stability-FPT-XP}
Let $\I$ be a strict (but not necessarily bipartite) preference system, and $Q$ any optimization problem 
where the task is to maximize or minimize some function $f$ over $\S(\I)$ such that
$f(M)$ can be computed in polynomial time for any matching~$M \in \S(\I)$. 
Then $Q$ can be solved 
\begin{itemize}
\item[(i)] \vspace{-4pt} in FPT time with parameter~$\distedge(\I)$ or~$\distswap(\I)$;
\item[(ii)] in polynomial time if $\distvert(\I)$ is constant.
\end{itemize}
\vspace{-4pt} 
In particular, these results hold for 
\myproblem{Sex-Equal Stable Matching}, \myproblem{Balanced Stable Matching}, \myproblem{(Generalized) Median Stable Matching}\footnote{%
Although the problem of finding a (generalized) median matching is not  an optimization problem over $\S(\I)$, 
it is clear that it can be solved in $|\S(\I)|\cdot O(|\I|)$ time.
},
\myproblem{Egalitarian Stable Roommates}, and \myproblem{Maximum-Weight Stable Roommates}.
\end{theorem}

\subsubsection{Tightness of our bounds on the number of stable matchings.}
\label{sec:tightness}

The bounds stated in Lemmas~\ref{lem:SM-enum-edge} and~\ref{lem:SM-enum-vert} are tight in the following sense. 
First, the tightness of Lemma~\ref{lem:SM-enum-edge} is shown by the existence of a strict preference system~$\I_k$ 
for any $k \in \mathbb{N}$ such that $\distedge(\I_k)=k$ and $\I_k$ admits $2^k$ stable matchings. 
A simple example for such an instance~$\I_k$ is obtained by taking $k$ disjoint copies of 
a preference system whose underlying graph is a cycle of length four, with cyclic preferences.

Second, to show the asymptotic tightness of Lemma~\ref{lem:SM-enum-vert}, one can construct a strict preference system~$\mathcal{J}_{k,n}$ 
for any $k,n \in \mathbb{N}$ with $n \geq k$ such that the following properties are satisfied: 
${\distvert(\J_{k,n})=k}$, the number of vertices in~$\mathcal{J}_{k,n}$ is $2n$, 
and $\J_{k,n}$ admits $\binom{n}{k}$ stable matchings. Such a family of preference systems is presented in Example~\ref{ex:1}.

\begin{example}
\label{ex:1}
Let the vertex set of~$\mathcal{J}_{k,n}$ be  $V=A \cup B \cup S$ 
where $A=\{a_1,\dots, a_n\}$, $B=\{ b_1, \dots, b_{n-k}\}$, and $S= \{s_1,\dots,s_k\}$; note $|V|=2n$.
Let $a_i \in A$ be adjacent to all vertices of~$S$, and let it be adjacent to some~$b_j \in B$ exactly if $i-k \leq j \leq i$;
we denote by $E$ the set of edges thus defined, and we let $G=(V,E)$ be the graph underlying $\mathcal{J}_{k,n}$.
See Figure~\ref{fig:ex1} for an illustration.
The preferences of the vertices are as follows: 
\begin{itemize}
\item any vertex~$b \in B$ prefers $a_j$ to~$a_i$ for some $a_j,a_i \in N_G(b)$ exactly if $j>i$;
\item any vertex $s \in S$ prefers $a_j$ to~$a_i$ for some $a_j,a_i \in N_G(s)$ exactly if $j<i$;
\item the preference list of some $a_i \in A$ is $b_i,s_1,b_{i-1},s_2, \dots, b_{i-k+1},s_k,b_{i-k}$; if some of these vertices are not defined, they are omitted from the list.
\end{itemize}
Observe that deleting the vertices of~$S$ from~$\mathcal{J}_{k,n}$ yields an instance that admits a master list, and moreover, 
deleting less than~$k$ vertices cannot result in an instance in~$\FamML$, since any vertex of~$S$ and any vertex of~$B$ disagree on the
order of vertices in~$A$. Hence, $\distvert(\mathcal{J}_{k,n})=k$. 

\begin{figure}[t]
\label{fig:ex1}
\begin{center}
\includegraphics[scale=1]{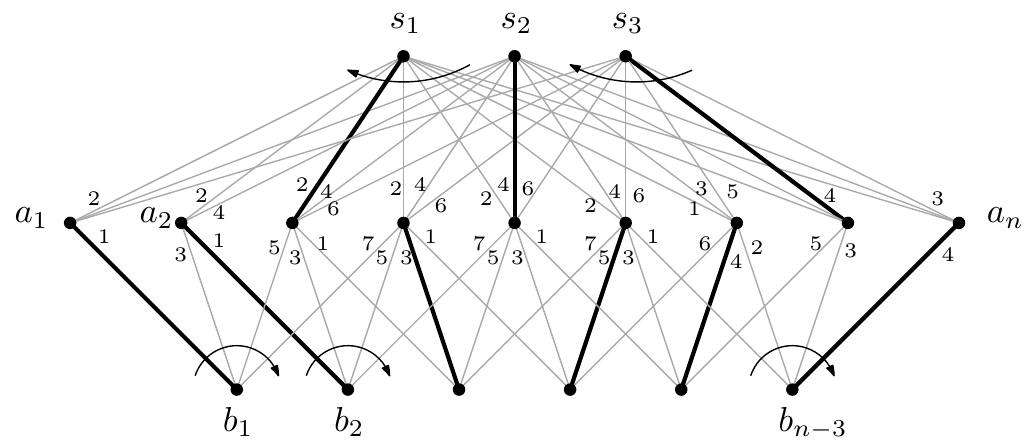}
\end{center}
\caption{Illustration of the preference system~$\mathcal{J}_{k,n}$ for $k=3$. Edges of the underlying graph are shown in grey, while edges of a stable matching are shown in black (and bold). Preferences are depicted either by numbers determining the rank of the given edge, or by an arrow, pointing from less-preferred edges towards more-preferred ones. }
\end{figure}

We claim that $\mathcal{J}_{k,n}$ admits $\binom{n}{k}$ stable matchings. 
Let $A'=\{a_{i_1},\dots, a_{i_k}\}$ be any subset of $k$  vertices from~$A$. 
We define a matching $M_{A'}$ as follows. First, let~$M_{A'}$ contain the edges $s_j a_{i_j}$ for each $j \in [k]$.
Second, let $M_{A'}$ match some~$b_j$ with $a_{j+\ell}$ if $\ell$ is the number of indices in $\{i_1,\dots,i_k\}$ smaller or equal to~$j$. 
It is straightforward to verify that $M_{A'}$ is stable, since no vertex of~$A$ can be part of a blocking edge. 
Since stable matchings in $\mathcal{J}_{k,n}$ are complete, any stable matching~$M$ must be of the form $M_{A'}$ for some $A' \subseteq A$ 
(namely, for ${A'=\{M(s):s \in S\}}$), 
so the number of stable matchings in $\mathcal{J}_{k,n}$ is exactly $\binom{n}{k}$.
\end{example}

We remark that we were not able to construct a preference system that, for a given~$k$, is at $\distswap$-distance~$k$ from 
admitting a master list, and contains $2^k$ stable matchings. Recall the preference system~$\I_k$ that served as an example to 
prove the tightness of Lemma~\ref{lem:SM-enum-edge} for the $\distedge$-distance, 
containing $k$ copies of a four-cycle with cyclic preferences. It is not hard to see that $\distswap(\I_k)=2k$, 
implying that the number of stable matchings in a strict preference system~$\I$ can be as large as $2^{\distswap(\I)/2}$. 
We leave it as an open question to determine the exact bound on the maximum number of 
stable matchings in a strict preference system~$\I$ as a function of~$\distswap(\I)$.

\subsection{Maximum-utility popular matchings with instability costs}
\label{sec:applic-pop}

We now turn our attention to the \MUPMU{} problem, studied in~\cite{CS-unstable-pop-arxiv}: 
given a strict preference system~$\I=(G,\prec)$, a utility function $\omega: E(G) \rightarrow \mathbb{N}$, 
a cost function $c: E(G) \rightarrow \mathbb{N}$,  an objective value $t \in  \mathbb{N}$ and a budget $\beta \in  \mathbb{N}$,
the task is to find a popular matching in~$\I$ whose utility is at least~$t$ and whose
blocking edges have total cost at most~$\beta$.
Our aim is to investigate whether we can solve this problem efficiently for instances that are close to admitting a master list.
 
Note that in general this problem is computationally hard even if the given preference system is strict, bipartite, admits a master list, 
and the cost and utility functions are very simple.  
Namely, given a strict, bipartite preference system $(G,\prec) \in \FamML$ for which a stable matching has size $|V(G)|/2-1$,
it is $\mathsf{NP}$-hard and $\mathsf{W}[1]$-hard with parameter~$\beta$ to find a complete popular matching 
(i.e., one that is larger than a stable matching) 
that admits at most~$\beta$ blocking edges~\cite{CS-unstable-pop-arxiv}. 
Nevertheless, if the total cost $\beta$ of the blocking edges that we allow is a constant and each edge has cost at least~1, 
then \MUPMU{} can be solved in polynomial time for bipartite, strict preference systems that admit a master list 
(in fact, it suffices to assume that the preferences of all vertices on one side of the bipartite 
input graph are consistent with a master list), 
representing an island of tractability for this otherwise extremely hard problem~\cite{CS-unstable-pop-arxiv}.
Therefore, it is natural to ask whether we can extend this result for strict preferences systems that are 
close to admitting a master list. Theorem~\ref{thm:MUPMU-almost-masterlist} answers this question affirmatively.

\begin{theorem}
\label{thm:MUPMU-almost-masterlist}
Let $\I$ be a strict (but not necessarily bipartite) preference system with $G=(V,E)$.
Then an instance $(\I,\omega,c,t,\beta)$ of \MUPMU{} where $c(e) \geq 1$ for all edges $e \in E$, and $\beta$ is constant
can be solved 
\begin{itemize}
\item[(i)] \vspace{-4pt} in FPT time with parameter~$\distedge(\I)$ or~$\distswap(\I)$;
\item[(ii)] in polynomial time if $\distvert(\I)$ is constant.
\end{itemize}
\end{theorem}

We apply the same approach as in Section~\ref{sec:enum-bounds},
with a crucial difference: for the algorithms proving Theorem~\ref{thm:MUPMU-almost-masterlist} 
we will need to determine a set of edges or vertices whose deletion yields an instance in~$\FamML$.
Using such a set, we then apply Lemma~\ref{lem:enum-fixedblocks-edge} or~\ref{lem:enum-fixedblocks-vert} below; 
these are generalizations of Lemmas~\ref{lem:SM-enum-edge} and~\ref{lem:SM-enum-vert} 
for the case when we allow a fixed set of edges to block the desired matching. 

\begin{mylemma}
\label{lem:enum-fixedblocks-edge}
Given a strict preference system~$\I=(G,\prec)$ with $G=(V,E)$ and edge sets $B \subseteq E$ and $S \subseteq E$
such that $\I-S \in \FamML$, the number of matchings~$M$ for which $B=\bp(M)$ is at most $2^{|S|}$, 
and it is possible to enumerate them in time~$2^{|S|} \cdot O(|E|)$. 
\end{mylemma}
\begin{proof}
We propose an algorithm that enumerates all matchings~$M$ of the input instance~$\I$ for which $\bp(M)=B$ 
by using a bounded search tree approach.
Let~$M$ be a hypothetical stable matching in~$\I$. 
The algorithm tries all possible ways to find~$M$, thus enumerating all stable matchings, as follows. 
See Algorithm~\ref{alg:enum-edge} for a pseudocode. 

First, the algorithm ``guesses'' the set $M_S=M \cap S$. 
Then it creates a preference system~$\I'$ by deleting the edges of~$B \cup S$ and also the vertices in~$V(M_S)$ from~$\I$;
note that $\I'$ admits a master list. 
Therefore, there exists a unique stable matching~$M'$ in~$\I'$ which can be found in time~$O(|E|)$~\cite{IMS08}\footnote{%
In fact, Irving et al.~\cite{IMS08} only considered the bipartite case, but it is easy to see that the straightforward algorithm 
where vertices pick their partners one-by-one according to their order on the master list, each choosing their favorite among those still available, results in the unique stable matching, regardless of whether the underlying graph is bipartite or not.
}.
The algorithm then checks if $\bp(M_S \cup M')=B$ holds in~$\I$ (this can be performed in~$O(|E|)$ time as well), 
and if so, outputs it. 
Observe that the algorithm outputs at most $2^{|S|}$ matchings, and has running time $2^{|S|} \cdot O(|E|)$, 
because there are $2^{|S|}$ possible ways to choose~$M_S$.

\begin{algorithm}
\caption{Enumerating all matchings in a preference system~$\I$  that admit a fixed set $B$ of
blocking edges,
given a set $S$ of edges in~$\I$ such that $\I-S \in \FamML$.}
\label{alg:enum-edge}
\begin{algorithmic}[1]
\ForAll{$M_S \subseteq S$} { 
	\State Let $\I'=\I-(S \cup B)-V(M_S)$. \Comment{Note $\I' \in \FamML$.}
	\State Compute the unique stable matching~$M'$ in~$\I'$.
	\State Let $M=M_S \cup M'$.
	\If{$\bp(M)=B$ in~$\I$} {\bf output} $M$.
	\EndIf
	}
\EndFor
\end{algorithmic}
\end{algorithm}

It is clear that any matching output by the algorithm satisfies $\bp(M)=B$ in~$\I$. 
It remains to prove that it enumerates all such matchings in~$\I$. So let $M$ be any matching of~$\I$ for which $\bp(M)=B$.
Consider the iteration corresponding to choosing $M_S=M \cap S$ on line~1 of Algorithm~\ref{alg:enum-edge}.
Note that $M \setminus M_S$ is a matching that is present in~\hbox{$\I'=\I-(S \cup B) -V(M_S)$}, 
and we show that it is stable in~$\I'$ as well. 
Note that by definition, $M$ is stable in~$\I-B$, and also in $\I-(B \cup (S \setminus M))$.
Furthermore, deleting the edges of $M_S \subseteq M$ together with their endpoints yields a stable matching in the remaining instance,
and so~$M \setminus M_S$ is also stable in $\I-(B \cup (S \setminus M))-V(M_S)=\I'$.
Since there is only one stable matching in~$\I'$, we obtain that $M \setminus M_S$ must be the matching~$M'$ 
found at line~3 (within the iteration corresponding to choosing~$M_S$ on line~1). 
Hence, the algorithm creates the matching $M_S \cup (M \setminus M_S)=M$ on the next line, and since $\bp(M)=B$ holds in~$\I$, 
the algorithm outputs it on line~5.
\qed
\end{proof}

\begin{mylemma}
\label{lem:enum-fixedblocks-vert}
Given a strict preference system~$\I=(G,\prec)$ with $G=(V,E)$, an edge set $B \subseteq E$, and a vertex set $S \subseteq V$
such that $\I-S \in \FamML$, the number of matchings~$M$ for which $B=\bp(M)$ is at most $|V|^{|S|}$, 
and it is possible to enumerate them in time~$|V|^{|S|} \cdot O(|E|)$. 
\end{mylemma}
\begin{proof}
We propose an algorithm that enumerates all matchings~$M$ of the input instance~$\I$ for which $\bp(M)=B$ 
by using a bounded search tree approach.
Let~$M$ be a hypothetical stable matching in~$\I$. 
The algorithm tries all possible ways to find~$M$, thus enumerating all stable matchings, as follows. 
See Algorithm~\ref{alg:enum-vert} for a pseudocode. 

First, the algorithm ``guesses'' the set $M_S \subseteq E$ of those edges in~$M$ that have at least one endpoint in~$S$. 
Then it creates a preference system~$\I'$ by deleting the edges of~$B$ and the vertices in~$V(M_S) \cup S$ from~$\I$;
note that $\I'$ admits a master list. 
Therefore, there exists a unique stable matching~$M'$ in~$\I'$ which can be found in time~$O(|E|)$~\cite{IMS08}.
The algorithm then checks if $\bp(M_S \cup M')=B$ holds in~$\I$ (this can be performed in~$O(|E|)$ time as well), 
and if so, outputs it. 
Observe that the algorithm outputs at most $|V|^{|S|}$ matchings, and has running time $|V|^{|S|} \cdot O(|E|)$, 
because there are at most $|V|^{|S|}$ possible ways to choose~$M_S$.

\begin{algorithm}
\caption{Enumerating all matchings in a preference system~$\I$ that admit a fixed set $B$ of
blocking edges, given a set $S$ of vertices in~$\I$ such that~$\I-S \in \FamML$.}
\label{alg:enum-vert}
\begin{algorithmic}[1]
\ForAll{matching $M_S$ in~$\I$ with each edge of~$M_S$ having an endpoint in~$S$} { 
	\State Let $\I'=\I-B - (S \cup V(M_S))$. \Comment{Note $\I' \in \FamML$.}
	\State Compute the unique stable matching~$M'$ in~$\I'$.
	\State Let $M=M_S \cup M'$.
	\If{$\bp(M)=B$ in~$\I$} {\bf output} $M$.
	\EndIf
	}
\EndFor
\end{algorithmic}
\end{algorithm}

It is clear that any matching output by the algorithm satisfies $\bp(M)=B$ in~$\I$. 
It remains to prove that it enumerates all such matchings in~$\I$. So let $M$ be any matching of~$\I$ for which $\bp(M)=B$.
Consider the iteration corresponding to choosing $M_S= \{e: e \in M, e \textrm{ has an endpoint in } S \}$ 
on line~1 of Algorithm~\ref{alg:enum-vert}.
Note that $M \setminus M_S$ is a matching that is present in~\hbox{$\I'=\I-B-S-V(M_S)$}, and we show 
that it is stable in~$\I'$ as well.
Note that $M$ is stable in~$\I-B$ by definition.
Furthermore, deleting the edges of~$M_S \subseteq M$ together with their endpoints
and the vertices of $S \setminus V(M_S)$, left unmatched by~$M$, 
yields a stable matching in the remaining instance, 
and so $M \setminus M_S$ is stable in~$\I-B-(V(M_S) \cup S)=\I'$.
Since there is only one stable matching in~$\I'$, we obtain that $M \setminus M_S$ must be the matching~$M'$ 
found at line~3 (within the iteration corresponding to choosing~$M_S$ on line~1). 
Hence, the algorithm creates the matching $M_S \cup (M \setminus M_S)=M$ on the next line, and since $\bp(M)=B$ holds in~$\I$, 
the algorithm outputs it on line~5.
\qed
\end{proof}

We now present our key lemma that, combined with Theorem~\ref{thm:MLED-2approx-FPT}, 
is the basis for Theorem~\ref{thm:MUPMU-almost-masterlist}.

\begin{mylemma}
\label{lem:MUPMU-almost-masterlist}
An instance $(\I,\omega,c,t,\beta)$ of \MUPMU{} where $\I=(G,\prec)$ is a strict preference system with underlying graph $G=(V,E)$
and $c(e) \geq 1$ for all edges $e \in E$, 
can be solved 
\begin{itemize}
\item[(i)] in~$2^{|S|} \cdot O(|E|^{\beta+1}\sqrt{|V|} \log |V|)$ time,
assuming that a set $S \subseteq E$ of edges such that $\I-S \in \FamML$ is provided as part of the input; \vspace{4pt}
\item[(ii)] in~$|V|^{|S|} \cdot O(|E|^{\beta+1}\sqrt{|V|} \log |V|)$ time,
assuming that a set $S \subseteq V$ of vertices such that $\I-S \in \FamML$ is provided as part of the input.
\end{itemize}
\end{mylemma}

\begin{proof}
Since all edges have cost at least~$1$, we know that the desired matching may admit at most~$\beta$ blocking edges. 
For each possible edge set $B \subseteq E$ of size at most~$\beta$ whose cost does not exceed the budget~$\beta$, 
the algorithm of Lemma~\ref{lem:enum-fixedblocks-edge} can  in~$2^{|S|} \cdot O(|E|)$ time enumerate all 
matchings~$M$ with $\bp(M)=B$ in case (i), and 
the algorithm of Lemma~\ref{lem:enum-fixedblocks-vert} can in~$|V|^{|S|} \cdot O(|E|)$ time enumerate all 
matchings~$M$ with $\bp(M)=B$ in case (ii).
For each such matching, we can check in linear time whether it has total utility at least~$t$, 
and we can test in $O(|E| \sqrt{|V|} \log |V|)$ time whether it is popular~\cite{BIM10,DPS18}. 
\qed
\end{proof}

We are now ready to prove Theorem~\ref{thm:MUPMU-almost-masterlist}.

\begin{proof}[of Theorem~\ref{thm:MUPMU-almost-masterlist}]
Let our input instance of \MUPMU{} be $\J=(\I,\omega,c,t,\beta)$
where $\I=(G,\prec)$ is a strict preference system with $G=(V,E)$
and $c(e) \geq 1$ for all edges $e \in E$.

First, suppose that $\distvert(\I)$ is a constant. 
Then we find a set~$S \subseteq V$ of vertices with $|S| =\distvert(\I)$ 
whose deletion yields a preference system admitting a master list;
this can be performed in $|V|^{\distvert(\I)}=|V|^{O(1)}$ time by trying all possible vertex sets of size at most~$\distvert(\I)$. 
Then we apply Lemma~\ref{lem:MUPMU-almost-masterlist} to solve our instance~$\J$, using the set~$S$. 
Note that the total running time is indeed polynomial, 
assuming that $\distvert(\I)$ as well as our budget~$\beta$ is constant.

Second, assume that our parameter is $\distedge(\I)$. To obtain an FPT algorithm with respect to~$\distedge(\I)$, 
we first use the 2-approximation of Theorem~\ref{thm:MLED-2approx-FPT}
to determine a set~$S \subseteq E$ of edges with $|S| \leq 2 \distedge(\I)$ 
whose deletion yields a preference system admitting a master list.
Given the set~$S$, we again apply Lemma~\ref{lem:MUPMU-almost-masterlist} to solve our instance~$\J$.
Observe that the total running time now is indeed FPT with parameter~$\distedge(\I)$, assuming that $\beta$ is constant.

Since $\distswap(\I) \geq \distedge(\I)$, an FPT algorithm with parameter~$\distedge(\I)$ immediately yields 
an FPT algorithm with parameter~$\distswap(\I)$. A more direct approach is to use Corollary~\ref{cor:MLS-FPT} 
to directly determine a set~$S$ of swaps of minimum cardinality for which $I \lhd S \in \FamML$, 
transform~$S$ into a set~$S'$ of edges with size~$|S'| \leq |S|$ such that  $\I-S \in \FamML$,
and then apply Lemma~\ref{lem:MUPMU-almost-masterlist} to solve~$\J$ using the edge set~$S'$.
\qed
\end{proof}

\section{Summary and further research}
\label{sec:conclusion}
We summarize our main results on MLS, MLED, and MLVD in Table~\ref{tab:summary}.
Interestingly, all our hardness results hold for strict preference systems, and
we were able to extend all our positive results for preference systems with weak orders.

\begin{table}[t]
\caption{Summary of our results on MLS, MLED, and MLVD. Results marked by the sign~$\dagger$ assume the Unique Games Conjecture.}
\label{tab:summary}
\centering
\begin{tabular}{l@{\hspace{4pt}}c@{\hspace{4pt}}c}
 problem &  parameterized complexity & approximation\\
\noalign{\hrule}
MLS & FPT wrt~$k$ (Cor.~\ref{cor:MLS-FPT}, Thm.~\ref{cor:MLS-ties-FPT}) & constant-factor approx. is $\mathsf{NP}$-hard 
(Thm.~\ref{thm:MLS-MLED-NPhard})$^{\dagger}$ \\[2pt]
MLED  & $\mathsf{W}[1]$-hard wrt~$k$ (Thm.~\ref{thm:MLED-W1hard}) & constant-factor approx. is $\mathsf{NP}$-hard 
(Thm.~\ref{thm:MLS-MLED-NPhard})$^{\dagger}$  \\
& & 2-approx. FPT alg wrt $k$ 
(Thms.~\ref{thm:MLED-2approx-FPT},~\ref{thm:MLED-ties-2approx-FPT}) \\[2pt]
MLVD & $\mathsf{W}[2]$-hard wrt~$k$ (Thm.~\ref{thm:MLVD-NPhard}) & $f(k)$-approx. is $\mathsf{W}[1]$-hard wrt $k$ (Thm.~\ref{thm:MLVD-NPhard})
\end{tabular}
\end{table}

There are a few questions left open in the paper. 
We gave asymptotically tight bounds on the maximum number of stable matchings in a strict preference
system $\I$ as a function of $\distedge(\I)$ and $\distvert(\I)$, but we were not able to do the same for $\distswap(\I)$.
Another question is whether the approximation factor of our 2-approximation FPT algorithm for MLED can be improved.

Apart from answering these specific questions, there are several possibilities for future research.
One direction would be to identify further problems that can be solved efficiently on preference systems 
that are close to admitting a master list. 
Also, it would be interesting to see how these measures vary in different real-world scenarios, 
and to find those practical applications where preference profiles are usually close to admitting a master list. 

\subsubsection{Acknowledgments.} 
We would like to thank the reviewers of WALCOM~2023 for their helpful comments.
This research was supported by the Hungarian Academy of Sciences under its Momentum Programme (LP2021-2) and the Hungarian Scientific Research Fund (OTKA grants~K128611 and K124171). 
\bibliographystyle{abbrv}
\bibliography{ml-long}

\end{document}